\documentclass{article}
\usepackage{amsmath}
\usepackage{amsfonts}
\usepackage{bm}
\usepackage{amsthm}
\usepackage{graphicx}
\usepackage{url,comment,xcolor}
\usepackage[normalem]{ulem}
\makeatletter
\def\squiggly{\bgroup \markoverwith{\textcolor{red}{\lower3.5\p@\hbox{\sixly \char58}}}\ULon}
\makeatother
\usepackage{bbm,mathrsfs}
\usepackage{pgfplots,tikz}
\pgfplotsset{compat=1.17}

\newtheorem{thm}{Theorem}[section]
\newtheorem{lem}[thm]{Lemma}
\newtheorem{pro}[thm]{Proposition}
\newtheorem{cor}[thm]{Corollary}

\newtheorem{defn}[thm]{Definition}
\newtheorem{rem}[thm]{Remark}
\usepackage[unicode=true,psdextra,
 bookmarks=true,bookmarksnumbered=true,bookmarksopen=false,
 breaklinks=true,pdfborder={0 0 1},backref=false,colorlinks=true]
 {hyperref}

\usepackage{authblk}

\begin{document}

\title{Momentum Space Feynman Integral for the Bound State Aharonov-Bohm Effect}

\author[1]{Alviu Rey Nasir\footnote{Corresponding author.  E-mail: alviurey.nasir@g.msuiit.edu.ph}}

\author[1]{Jingle Magallanes}

\author[2]{Herry~Pribawanto~Suryawan}

\author[3]{Jos\'{e} Lu\'{i}s Da Silva}

\affil[1]{Department of Physics, College of Science and Mathematics\newline \& Premier Research Institute of Science and Mathematics\newline Mindanao State University-Iligan Institute of Technology\newline Iligan City 9200, Philippines}
\affil[2]{Department of Mathematics, Faculty of Science and Technology\newline Sanata Dharma University, Yogyakarta 55283, Indonesia}
\affil[3]{CIMA Faculdade de C{\^e}ncias Exatas e da Engenharia\newline Campus Universit{\'a}rio da Penteada, Universidade da Madeira\newline 9020-105 Funchal, Portugal}
\maketitle

\begin{abstract}
We construct the Feynman integral for the Schr\"{o}dinger propagator in the polar conjugate momentum space, which describes the bound state Aharonov-Bohm effect, as a well-defined white noise functional.
\end{abstract}



\section{Introduction}	

In quantum mechanics, Feynman integrals have been shown to produce quantum propagators that solve the corresponding Schr\"{o}dinger equation for the system (see, e.g., Ref.~\cite{FeHi65}).  As a mathematically rigorous framework justifying the heuristic formulation of the Feynman integrals, White Noise Analysis (WNA) \cite{Hid75} has become a handy tool for a physicist to derive quantum propagators for select quantum systems in either the coordinate space or the momentum space; see, e.g., Refs.~\cite{HS83, Bock2013}.

Regarding the Feynman path integrals in the polar coordinates, Peak and Inomata first published their expression in terms of the Hamiltonian that contains the conjugate momenta of the polar coordinates \cite{peak1969summation}.  Gerry and Singh \cite{gerry1979feynman} were the first to investigate the Feynman path integral approach to the Aharonov-Bohm (AB) effect, a quantum mechanical phenomenon initially established by Aharonov and Bohm \cite{aharonov1959significance}. Using WNA, Bernido and Bernido \cite{BB02} derived the propagators for the particle on a circle and with the bound state AB potential in the coordinate space.  However, in these cases, they did not consider and obtain propagators in the conjugate momentum space.  On the other hand, the case of the momentum space AB effect (but not the bound state) in the Cartesian coordinates has already been considered (see, e.g., Ref.~\cite{dragoman2008aharonov}).  The case of the bound state AB effect in the conjugate momentum space has not yet been considered elsewhere.  

Inspired by the advantageous role of momentum variables in experimental design, this paper focuses on a quantum system for the bound state AB effect \cite{aharonov1959significance, kretzschmar1965aharonov} in the conjugate momentum space using the WNA framework. We show that the propagator reduces to that for the case of the particle on a circle when the magnetic flux is switched off. 

The paper is organized as follows. In Section~\ref{sec:WNA}, we give a brief review of the WNA.  In Section~\ref{sec:Feyn-AB-effect}, we describe the quantum system considered and define the Feynman integral for the system in the WNA framework. We derive its propagator and show that it solves the corresponding Schr\"{o}dinger equation in the conjugate momentum space.  In Section~\ref{sec:Alternative}, we write an alternative result within the WNA framework as a perturbation on a particle-on-a-circle system employing experimentally confirmed values.  
Finally, in Section~\ref{sec:conclusions}, we conclude and give a brief discussion of the results obtained.  

\section{White Noise Analysis}\label{sec:WNA}
In this section, we give a rather short introduction to WNA; see, e.g., Refs.~\cite{HKPS93, Bock2013} for more details. The starting point is the basic nuclear triple
\begin{equation}
    S_d \subset L_d^2 \subset S_d',
\end{equation}
where $L_d^2$ is the real separable Hilbert space of square integrable functions with norm 
\[
|f|_0^2 := \sum_{j=1}^d \int_{\mathbb{R}} f_j^2 (s)\, \mathrm{d}s, \quad f \in L_d^2 
\]
and $S_d$ is the nuclear space of $d$-dimensional vectors where each component is a Schwartz test function. $S_d'$ is the topological dual of $S_d$, the so-called vector-valued  tempered distributions. The dual pairing between $S'_d$ and $S_d$ is denoted by $\langle\cdot,\cdot\rangle$. The space $S'_d$  is provided with the $\sigma$-algebra of the cylinder sets $\mathscr{B}$. We define the Gaussian measure on the measurable space $(S'_d,\mathscr{B})$ using Minlos' theorem by fixing the characteristic function, 
\[
\int_{S'_d}\exp(\mathrm{i}\langle\omega,\varphi\rangle)\,\mathrm{d}\mu(\omega)=\exp\left(-\frac{1}{2}\langle\varphi,\varphi\rangle\right),\quad \varphi\in S_d.
\]
The triple $(S'_d,\mathscr{B},\mu)$ is called the vector-valued white noise space. Let $L^2(\mu):=L^2(S'_d,\mathscr{B},\mu;\mathbb{C})$ be the complex Hilbert space with a scalar product 
\[
(\!(F,G)\!) = \int_{S_d'} \bar{F}(\omega) G(\omega) \, \mathrm{d}\mu (\omega),\quad F,G\in L^2(\mu).
\]
Choosing a special subspace $(S_d)^1$ of test functionals, we construct a Gelfand triple around $L^2(\mu)$
\begin{equation}
    (S_d)^1 \subset L^2(\mu) \subset (S_d)^{-1}.
\end{equation}
The elements of the space $(S_d)^{-1}$ are called generalized white noise functionals or Kondratiev distributions, and the well-known Hida distributions $(S_d)^*$ (or generalized Brownian functionals) form a subspace, that is, $L^2(\mu) \subset (S_d)^* \subset (S_d)^{-1}$.  The dual pairing between $(S_d)^{-1}$ and $(S_d)^1$ is given as a bilinear extension of the inner product in $L^2(\mu)$ and is denoted by $\langle\!\langle\cdot,\cdot\rangle\!\rangle_\mu$. Distributions can be characterized by the so-called $T$-transform; see Theorem~\ref{Th1} below. Given $\Phi\in(S_d)^{-1}$, there exists $p,q\in\mathbb{N}_0$ such that we may define
\begin{equation}
    T\Phi(\xi) := \left\langle\!\! \left\langle \Phi, \mathrm{e}^{\mathrm{i}\langle \cdot , \xi \rangle} \right\rangle\!\! \right\rangle_\mu,\quad \xi\in U_{p,q}:=\big\{\xi\in S_d\mid |\xi|_p^2<2^{-q}\big\}. 
\end{equation}
It should be noted here that the definition of the $T$-transform can be extended via analytic continuation to the complexification of $S_d$ which is denoted by $S_{d,\mathbb{C}}$.

\begin{defn}\label{def:holmorphy}
	Let $U\subseteq S_{d,\mathbb{C}}$ be an open set and $F:U\longrightarrow\mathbb{C}$ a given function. Then $F$ is holomorphic on $U$ iff for all $\varphi_0\in U$
		\begin{itemize}
			\item[1.] for any $\varphi\in S_{d,\mathbb{C}}$ the map $\mathbb{C}\ni z\mapsto F(\varphi_0+z\varphi)\in\mathbb{C}$ is holomorphic in a neighborhood of zero in $\mathbb{C}$,
			\item[2.] there exists an open neighborhood $U'$ of $\varphi_0$ such that $F$ is bounded on $U'$.
		\end{itemize}
	$F$ is holomorphic at zero iff $F$ is holomorphic in a neighborhood of zero.
\end{defn}
\begin{thm}\label{Th1}
	Let $U\subseteq S_{d,\mathbb{C}}$ be an open set and $F:U\rightarrow \mathbb{C}$ be holomorphic at zero, then there exists a unique $\Phi\in(S_d)^{-1}$ such that $T\Phi=F$. Conversely, given $\Phi\in (S_d)^{-1}$, then $T\Phi$ is holomorphic at zero.  The correspondence between $F$ and $\Phi$ is bijective if we identify holomorphic functions that coincide on an open neighborhood of zero.
\end{thm}

In applications, we have to handle the convergence of sequences of distributions from $(S_d)^{-1}$ as well as integrals of $(S_d)^{-1}$-valued functions. The following two corollaries are a consequence of Theorem~\ref{Th1} and will be applied in what follows.
	\begin{cor}\label{seq}
		Let $\left(\Phi_n\right)_{n\in\mathbb{N}}\subset(S_d)^{-1}$ be a sequence such that there exists $U_{p,q}\subset S_d$, $p,q\in\mathbb{N}_0$, so that
		\begin{itemize}
			\item[1.] all $T\Phi_n$ are holomorphic on $U_{p,q}$,
			\item[2.] there exists a $C>0$ such that $\vert T\Phi_n(\varphi)\vert \leq C$ for all $\varphi\in U_{p,q}$ and all $n\in\mathbb{N}$,
			\item[3.] $\left(T\Phi_n(\varphi)\right)_{n\in\mathbb{N}}$ is a Cauchy sequence in $\mathbb{C}$ for all $\varphi\in U_{p,q}$.
		\end{itemize}
	   Then $\left(\Phi_n\right)_{n\in\mathbb{N}}$ converges strongly in $(S_d)^{-1}$.
	\end{cor}
	
	\begin{cor}\label{Bochner}
		Let $( \Lambda, \mathscr{A}, \nu )$ be a measure space and $\lambda \mapsto \Phi_\lambda$ a mapping from $\Lambda$ to $\left(S_d\right)^{-1}$.  We assume that there exist $U_{p,q}\subset S_d$, $p,q \in \mathbb{N}_0$, such that
		\begin{itemize}
			\item[1.] $T\Phi_\lambda$ is holomorphic on $U_{p,q}$ for every $\lambda\in\Lambda$,
			\item[2.] the mapping $\lambda \mapsto T\Phi_\lambda(\varphi)$ is measurable for every $\varphi\in U_{p,q}$,
			\item[3.] there exists $C\in L^1(\Lambda,\mathscr{A},\nu)$ such that
			\begin{equation}
				\vert T\Phi_\lambda(\varphi)\vert \leq C(\lambda)
			\end{equation}
			for all $\varphi\in U_{p,q}$ and for $\nu$-almost all $\lambda\in \Lambda$.
		\end{itemize}
		Then $\Phi_\lambda$ is Bochner integrable.  In particular,
		\begin{equation}
			\int_\Lambda \Phi_\lambda\, \mathrm{d}\nu (\lambda)\in \left(S_d\right)^{-1}
		\end{equation}
		and we may interchange dual pairing and integration
			\begin{equation}
				\left\langle \!\!\!\left\langle \int_{\Lambda}\Phi_{\lambda}\,\mathrm{d}\nu(\lambda),\xi\right\rangle \!\!\!\right\rangle_\mu=\int_{\Lambda}\langle\!\langle\Phi_{\lambda},\xi\rangle\!\rangle_\mu\,\mathrm{d}\nu(\lambda),\quad\xi\in(S_{d})^{1}.
			\end{equation}
	\end{cor}

\begin{lem}[cf.~Lemma 2.17 in Ref.~\cite{Bock2013a}]
    Let ${K}$ be a $d \times d$ block operator matrix on $L^2_{d,\mathbb{C}}$, such that ${\bm{N}}:={\mathrm{Id}} + {K}$ is bounded with bounded inverse.  Furthermore, assume that $\det ({\mathrm{Id}}+{K})$ exists and is different from zero.  Let $M_{\bm{N}^{-1}}$ be the matrix given by an orthogonal system $(\eta_k)_{k=1,\dots,J}$, $J\in \mathbb{N}$, of nonzero functions of $L^2_{d,\mathbb{C}}$  in the bilinear form $(\cdot, {\bm{N}}^{-1} \cdot)$, that is, $(M_{{\bm{N}}^{-1}})_{i,j}=(\eta_i, {\bm{N}}^{-1}\eta_j), \;\; 1\leq i, j\leq J$.  Under the assumption that for every $f\in L^2_{d,\mathbb{C}}$ either
			\begin{equation}
			\big(f,\Re(M_{\bm{N}^{-1}})f\big)_0>0 \quad \mathrm{or} \quad \Re(M_{\bm{N}^{-1}})=0 \quad \mathrm{and} \quad \big(f,\Im(M_{\bm{N}^{-1}})f\big)_0 \neq 0,
            \label{conditionofthelemma}
			\end{equation}
			where $M_{\bm{N}^{-1}} = \Re(M_{\bm{N}^{-1}}) + \mathrm{i}\Im(M_{\bm{N}^{-1}})$ with real matrices $\Re(M_{\bm{N}^{-1}})$ and $\Im(M_{\bm{N}^{-1}})$, then
			\begin{eqnarray}
			\Phi_{{ K}} := \mathrm{N}\exp(-\frac{1}{2}\langle \cdot, {K} \cdot \rangle)\cdot \exp(\mathrm{i}\langle \cdot, g \rangle ) \prod_{k=1}^J \delta_0(\langle \cdot, \eta_k \rangle - y_k),
			\end{eqnarray}
			for $g \in L^2_{d,\mathbb{C}}$, $ t>0$, $y_k \in \mathbb{R}, k=1,\dots,J$, exists as a Hida distribution.
			
Moreover for $f\in S_d$
\begin{multline}
T\Phi_{{K}}(f)=\frac{1}{\sqrt{(2\pi)^J \det(M_{\bm{N}^{-1}})}} \frac{1}{\sqrt{\det({\mathrm{Id}}+{K}))}} \\
\times \exp\left(-\frac{1}{2}\langle {f+g}, ({\mathrm{Id}}+{K})^{-1}{(f+g)}\rangle\right) \exp\left(\frac{1}{2}(u^\top (M_{\bm{N}^{-1}}^{-1})u )\right), \label{formula}
\end{multline}
where
\begin{equation}
			u=\big(\mathrm{i}y_1+(\eta_1,\bm{N}^{-1}(f+g)),\dots,\mathrm{i}y_J+(\eta_J,\bm{N}^{-1}(f+g)\big).
			\end{equation}
   \label{lemma}
		\end{lem}

\begin{rem}[cf.~Remark 2.13 in Ref.~\cite{Bock2013a}]\label{normalized}
     In case of a divergent or undefined factor $1/\sqrt{\det ({\mathrm{Id}}+{K})}$, the normalization constant $\mathrm{N}$ can be chosen properly to ``eliminate'' the factor in Eq.~\eqref{formula}.
\end{rem}

\section{The Feynman Integrals for the Bound State AB Effect}
\label{sec:Feyn-AB-effect}
Consider an illustration of the bound state AB system as depicted in Figure 1.

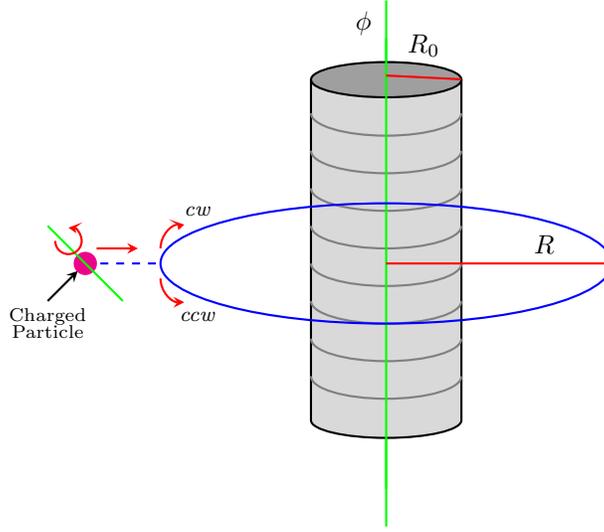
\begin{figure}[ht] 
\caption{\label{solenoid1}An illustration of the paths taken by a spinning electron around the solenoid from which came the magnetic flux ``inducing'' Aharonov-Bohm effect.}
\centering
 \centering 
\usetikzlibrary {shapes.geometric,arrows.meta} 
\begin{tikzpicture}[thick,>=stealth]  
\draw[green] (0,-3) -- (0,3); 
\node[cylinder,
    draw,
    cylinder uses custom fill,
    cylinder body fill = gray!30,
    cylinder end fill = gray!75,    
    aspect=2, 
    minimum height=5cm,
    minimum width=2cm,
    rotate=90]{};
\draw [gray](-1,2.0) arc [start angle=-180, end angle=0, x radius=1cm, y radius=.3cm];
\draw [gray](-1,1.5) arc [start angle=-180, end angle=0, x radius=1cm, y radius=.3cm];
\draw [gray](-1,1.0) arc [start angle=-180, end angle=0, x radius=1cm, y radius=.3cm];
\draw [gray](-1,.5) arc [start angle=-180, end angle=0, x radius=1cm, y radius=.3cm];
\draw [gray](-1,0) arc [start angle=-180, end angle=0, x radius=1cm, y radius=.3cm];
\draw [gray](-1,-.5) arc [start angle=-180, end angle=0, x radius=1cm, y radius=.3cm];
\draw [gray](-1,-1.0) arc [start angle=-180, end angle=0, x radius=1cm, y radius=.3cm];
\draw [gray](-1,-1.5) arc [start angle=-180, end angle=0, x radius=1cm, y radius=.3cm];
\draw[green] (0,-3.5) -- (0,3.5); 
\draw[arrows={->[scale=.8]},red](-3.85,.2)--(-3.3,.2);
\draw[arrows={->[slant=-.9,scale=.8]},red](-3,.2) arc [start angle=180, end angle=90, radius=9pt];
\draw (-2.5,.7) node {\footnotesize \textit{cw}}; 
\draw[arrows={->},red](-3,-.2) arc [start angle=-180, end angle=-90, radius=9pt];
\draw (-2.5,-.7) node {\footnotesize \textit{ccw}}; 
\draw[->](-4.5,-.5)--(-4.1,-.1);
\draw (-4.5,-.7) node {\scriptsize Charged};
\draw (-4.5,-.9) node {\scriptsize Particle};
\filldraw [magenta] (-4,0) circle [radius=4pt]; 
\draw (-.3,3.2) node {$\phi$}; 
\draw [red] (0,2.5)--(1,2.45); 
\draw (.5,2.9) node {$R_0$}; 
\draw[blue,rotate=90] (0, 0) ellipse (.8cm and 3cm);
\draw[red](0,0)--(3,0) node [above, pos=.70, above,color=black] {$R$};
\draw[green](-4.5,.5)--(-3.5,-.5);

\draw [arrows = {->[slant=.6]},red] (-4.4,.3) arc [start angle=-180, end angle=90, radius=5pt];
\draw[blue,style=dashed](-3.8,0)--(-3,0);
\end{tikzpicture}
\end{figure}

The circular path taken by a charged particle of charge $-e$ has a radius $R$ that is distant from the radius $R_0$ of the solenoid from which the magnetic flux $\phi$ comes.  The charged particle can either go \textit{cw} or \textit{ccw} with a conjugate momentum $p_{\theta,0}$ starting on the angular position $\theta_0$ at time $t_0 \geq 0$, and will gain a conjugate momentum of $p_\theta$ on an angular position $\theta$ at a later time $t > t_0$.  In a later calculation, we shall pin the final conjugate momentum at $p_{\theta}'$.  Note that the charged particle does not pass through a region where the magnetic flux is present, nor does it penetrate the solenoid.  In other words, the magnetic flux is in the excluded region.

The AB potential for the bound state has the form (See, e.g., Refs.~\cite{gerry1979feynman, BB02, kretzschmar1965aharonov, peshkin1990aharonov})
\begin{equation}
    V_{\mathrm{AB}} := - \frac{e\phi}{2\pi \hbar c} \dot{\theta} = \alpha \dot{\theta}, \label{V_AB}
\end{equation}
where $\hbar = h/(2\pi)$ with $h$ being the Planck's constant, $c$ is the speed of light in vacuum, $\alpha := -e\phi/(2\pi \hbar c)$ is the magnetic flux parameter, and $\dot{\theta} := \mathrm{d}\theta / \mathrm{d}t$ is the angular velocity.

In the polar coordinates, $\Vec{r} = (r, \theta)$, the Hamiltonian of the system is
\begin{eqnarray}
    H &:=& H(r,\theta,\dot{\theta},t,t_0)  \nonumber\\
    &=& \frac{1}{2}m_0 v^2 + V_{\mathrm{AB}} = \frac{1}{2}m_0 (\dot{r}^2 + r^2 \dot{\theta}^2 ) + V_{\mathrm{AB}}, \nonumber\\
    &=& \frac{1}{2}m_0 R^2 \dot{\theta}^2 + V_{\mathrm{AB}}, \label{hamiltonian}
\end{eqnarray}
where $m_0$ is the mass of the particle.  With $p_\theta := m_0 R\dot{\theta}$, the momentum conjugated to the angular position $\theta$, we have
\begin{equation}
    H = C_1p_\theta^2 - C_2p_\theta,\quad C_1:=\frac{1}{2m_0},\quad C_2:=\frac{-\alpha}{m_0 R}. \label{H}
\end{equation}

From Refs.~\cite{Bock2013, BB02}, we model the parametric expressions for $\theta$ and $p_\theta$ as an adaptation therein but with slight modifications,
\begin{equation}
    \theta (s) = \frac{p_{\theta,0}}{m_0R}(s-t_0+a) +  \langle\omega_\theta, E\mathbbm{1}_{[t_0,s)} \rangle, \quad E:=\sqrt{\frac{\hbar}{m_0 R^2}}, \; a \in \mathbb{R} \setminus \{0\},\; |a| \leq t \label{theta}
\end{equation}
and
\begin{equation}
    p_\theta (s) = p_{\theta,0} + D \omega_\theta(s),\quad D:=\sqrt{\hbar m_0}, \qquad 0 \leq t_0\leq s\leq t, 
    \label{p}
\end{equation}
which are chosen to conform with the Heisenberg uncertainty principle: we want to do a ``measurement'' (i.e., in this case, obtain a propagator) in the conjugate momentum space,  therefore we cannot ascertain the initial value of $\theta$ (that's why $\theta$ is expressed in terms of $p_{\theta,0}$ and not of $\theta_0$).  Furthermore, $a$ is chosen to physically mean that as the electron ``enters'' the circular path before $t_0$, it will have to ``pick up'' a certain momentum and so at $t_0$ it is already at a certain angular position $(p_{\theta,0}a)/(m_0R)$. 

We model the action of the classical paths as an adaptation of Refs.~\cite{Bock2013, peak1969summation, Kl90}
\begin{equation}
    S = -\int_{t_0}^t \big(R\theta(s)\dot{p_\theta}(s) + H(s) 
    \big)\, \mathrm{d}s. \label{action}
\end{equation}
Note that with Eqs.~\eqref{theta} and \eqref{p} (in the variable $\omega=(\omega_\theta,\omega_p)$)
\begin{equation*}
\int_{t_0}^t \frac{\mathrm{d}}{\mathrm{d}s}\big(\theta(s)p_{\theta}(s)\big) \,\mathrm{d}s = \int_{t_0}^t \dot{\theta}(s)p_{\theta}(s) \,\mathrm{d}s + \int_{t_0}^t \theta(s)\dot{p}_{\theta}(s) \,\mathrm{d}s.
\end{equation*}
From which follows
\begin{multline}
\int_{t_0}^t \theta(s)\dot{p}_{\theta}(s) \,\mathrm{d}s = \int_{t_0}^t \frac{\mathrm{d}}{\mathrm{d}s}\big(\theta(s)p_{\theta}(s)\big) \,\mathrm{d}s - \int_{t_0}^t \dot{\theta}(s)p_{\theta}(s) \,\mathrm{d}s \\
= \theta(t)p_\theta(t) - \theta(t_0)p_\theta(t_0) - \int_{t_0}^t \dot{\theta}(s)p_{\theta}(s) \,\mathrm{d}s \\
=\left(\frac{p_{\theta,0}}{m_0R}(t-t_0+a) + \big<\omega_\theta,E\mathbbm{1}_{[t_0,t)}\big>\right)p_\theta(t) - \frac{p_{\theta,0}}{m_0R}ap_\theta(t_0) - \int_{t_0}^t \dot{\theta}(s)p_{\theta}(s) \,\mathrm{d}s \\
= \frac{p_{\theta,0}}{m_0R} \big[ p_\theta(t)(t-t_0) + a \left( p_\theta(t) - p_{\theta,0}\right) \big] + \big<\omega_\theta,Ep_\theta(t)\mathbbm{1}_{[t_0,t)}\big> - \int_{t_0}^t \dot{\theta}(s)p_{\theta}(s) \,\mathrm{d}s,
\end{multline}
where we have set $\omega_\theta(t_0) = 0$,
which implies that the starting momentum
is fixed at $p_{\theta,0}$ on $t_0$.
Therefore, the action $S$ in Eq.~\eqref{action} becomes
\begin{multline}
S(\omega) = \int_{t_0}^t \big(-R\theta(s)\dot{p}_{\theta}(s) - H(s)\big)\,\mathrm{d}s\\
= -\frac{p_{\theta,0}}{m_0} \left[ p_\theta(t)(t-t_0) + a \left( p_\theta(t) - p_{\theta,0}\right) \right] - \big<\omega_\theta,REp_\theta(t)\mathbbm{1}_{[t_0,t)}\big> \\
 + R\int_{t_0}^t \dot{\theta}(s)p_{\theta}(s) \,\mathrm{d}s - \int_{t_0}^t H(s)\,\mathrm{d}s \\
= -\frac{p_{\theta,0}}{m_0} \left[ p_\theta(t)(t-t_0) + a \left( p_\theta(t) - p_{\theta,0}\right) \right] - \big<\omega_\theta,REp_\theta(t)\mathbbm{1}_{[t_0,t)}\big> \\
 + \frac{1}{m_0}\int_{t_0}^t p^2_{\theta}(s) \,\mathrm{d}s - \int_{t_0}^t H(s)
 \,\mathrm{d}s. \label{SintermsofH}
\end{multline}

With the Hamiltonian $H$ in Eq.~\eqref{hamiltonian}  and  $\omega=(\omega_\theta,\omega_p)$, we have
\begin{multline}
S(\omega) \\
= -\frac{p_{\theta,0}}{m_0} \left[ \left(p_\theta(t)-\frac{1}{2}p_{\theta,0}\right)(t-t_0) + a \left( p_\theta(t) - p_{\theta,0}\right) \right] - \sqrt{\frac{\hbar}{m_0}}\langle\omega_\theta,(p_\theta(t)-p_{\theta,0})\mathbbm{1}_{[t_0,t)}\rangle \\
 + \frac{\hbar}{2}\int_{t_0}^t \omega^2_\theta(s) \,\mathrm{d}s  -\frac{\alpha}{m_0R} p_{\theta,0} (t-t_0) - \sqrt{\frac{\hbar}{m_0}}\frac{\alpha}{R}\langle\omega_\theta, \mathbbm{1}_{[t_0,t)} \rangle \\ 
= -\frac{p_{\theta,0}}{2m_0} \left[ \left(p_\theta(t) + \frac{2\alpha}{R}\right)(t-t_0) +  \left( p_\theta(t) - p_{\theta,0}\right)(t-t_0 + 2a) \right]  \\
 + \frac{\hbar}{2}\int_{t_0}^t \omega^2_\theta(s) \,\mathrm{d}s + \sqrt{\frac{\hbar}{m_{0}}}\big<\omega_{\theta},\big(-\alpha R^{-1}-(p_{\theta}(t)-p_{\theta,0})\big)\mathbbm{1}_{[t_{0},t)}\big>.
 \label{S}
\end{multline}

\subsection{Incomplete Single Winding Case} \label{case1}
Consider now the case where the charged particle is constrained or guided to move in less than a single winding, either \textit{cw} (with $p_{\theta} < 0$) or \textit{ccw} (with $p_{\theta} > 0$).

The Feynman integrand $I_{V_{\mathrm{AB}}}$ in terms of white noise variables $\omega = (\omega_\theta, \omega_p)$ is given by (cf. Ref.~\cite{bock2014hamiltonian})
\begin{eqnarray}
    I_{V_{\mathrm{AB}}} &:=& \mathrm{N} \exp \Biggl[ {\frac{\mathrm{i}}{\hbar} S(\omega)} \Biggr] \exp \Biggl[ {\frac{1}{2}\int_{t_0}^t (\omega^2_\theta(s) + \omega^2_p(s) )\, \mathrm{d}s} \Biggr] \nonumber\\
    &  & \times \delta \Big( 
    p_{\theta,0}  +\frac{D}{t-t_0}\langle(\omega_{\theta},\omega_{p}),(\mathbbm{1}_{[t_0,t)},0)\rangle - p_\theta' \Big),
    \label{IAB}
\end{eqnarray}
where $\mathrm{N}$ is a normalization constant, the second exponential is a term to compensate for the Gaussian fall-off, and the delta function is to pin the end point at $p_\theta'$. 
 It is important to note that, in the delta function, we match the indicator function $\mathbbm{1}_{[t_0,t)}$ with the $\omega_\theta$ based on our conjugate momentum expression, Eq.~\eqref{p} (cf. Ref.~\cite{bock2014hamiltonian}, wherein it was matched with $\omega_p$ based on the momentum expression therein).

Substitution of Eq.~\eqref{S} into Eq.~\eqref{IAB} yields,
\begin{eqnarray}
I_{V_{\textrm{AB}}} & = & \mathrm{N}\exp\left\{ -\frac{\textrm{i}p_{\theta,0}}{2\hbar m_{0}}\left[\left(p_{\theta}(t)+\frac{2\alpha}{R}\right)(t-t_{0})+\left(p_{\theta}(t)-p_{\theta,0}\right)(t-t_{0}+2a)\right]\right\} \nonumber\\
 &  & \cdot\exp\left\{ \frac{\mathrm{i}}{2}\int_{t_{0}}^{t}\omega_{\theta}^{2}(s)\,\mathrm{d}s+\mathrm{i}C\big<\omega_{\theta},\mathbbm{1}_{[t_{0},t)}\big>\right\} \nonumber\\
 &  & \cdot\exp\Biggl[\frac{1}{2}\int_{t_{0}}^{t}(\omega^2_{\theta}(s)+\omega^2_{p}(s))\,\mathrm{d}s\Biggr] \delta \Big( 
    p_{\theta,0}  +\frac{D}{t-t_0}\langle(\omega_{\theta},\omega_{p}),(\mathbbm{1}_{[t_0,t)},0)\rangle - p_\theta' \Big) \nonumber\\ 
 & = & \mathrm{N}\exp\left\{ -\frac{\textrm{i}p_{\theta,0}}{2\hbar m_{0}}\left[\left(p_{\theta}(t)+\frac{2\alpha}{R}\right)(t-t_{0})+\left(p_{\theta}(t)-p_{\theta,0}\right)(t-t_{0}+2a)\right]\right\} \nonumber\\
 &  & \cdot\exp\left\{ -\frac{1}{2}\langle\omega,K\omega\rangle+\mathrm{i}\langle \omega,g\rangle\right\} \delta \Big( 
    p_{\theta,0}  +\frac{D}{t-t_0}\langle(\omega_{\theta},\omega_{p}),(\mathbbm{1}_{[t_0,t)},0)\rangle - p_\theta' \Big), \nonumber\\
 \label{I}
\end{eqnarray}
where $\langle\omega_{\theta},\delta_{t}\rangle$, $t\ge0$, exists in the sense of Corollary~\ref{Bochner}, the matrix $K$ has the form  
\begin{equation}
    K := \begin{pmatrix}
-\mathbbm{1}_{[t_0,t)} - \mathrm{i}\mathbbm{1}_{[t_0,t)} & -\mathbbm{1}_{[t_0,t)} \\
\mathbbm{1}_{[t_0,t)} & -\mathbbm{1}_{[t_0,t)}
\end{pmatrix},
\end{equation}
and $g$, $\eta_t$ are given by
\begin{equation}\label{eq:g}
    g:=(g_{{\theta}},0)=C\begin{pmatrix}\mathbbm{1}_{[t_0,t)}, & 0\end{pmatrix},\qquad C:=\frac{1}{\sqrt{\hbar m_0}}\left(\frac{-\alpha}{R} - (p_\theta(t)-p_{\theta,0})\right),
\end{equation}
\begin{equation}\label{eq:eta}
    \eta_t := \frac{D}{t-t_0}\left( \mathbbm{1}_{[t_0,t)},\, 0\right).
\end{equation}

Having our notations and calculations based on Lemma~\ref{lemma} and Remark~\ref{normalized} above, 
we particularly have, with a small perturbation $\varepsilon > 0$ (to meet the conditions in Eq.~\eqref{conditionofthelemma} of Lemma \ref{lemma}, as demonstrated in Refs.~\cite{Bock2013, bock2014hamiltonian})
\begin{equation}
    \bm{N}^{-1}_\varepsilon := (\mathrm{Id} +K)^{-1}_\varepsilon
    = \begin{pmatrix}
		\mathbbm{1}_{[t_0,t)^c}  & 0 \\
0 & \mathbbm{1}_{[t_0,t)^c} 
	\end{pmatrix} + \begin{pmatrix}
		\varepsilon & \mathbbm{1}_{[t_0,t)} \\
-\mathbbm{1}_{[t_0,t)} &  - \mathrm{i}\mathbbm{1}_{[t_0,t)}
	\end{pmatrix},
\end{equation}
and so,
\begin{equation}
    M_{\bm{N}_\varepsilon^{-1}}=
    \left(\eta_t, \bm{N}^{-1}_\varepsilon \eta_t \right)= \frac{\varepsilon D^2}{t-t_0} = \frac{\varepsilon \hbar m_0}{t-t_0}.
    \label{M}
\end{equation}

For any $f,g\in L^2_2(\mathbb{R})$ of the form $f=(f_\theta,f_{p_\theta})$ and $g=(g_\theta,g_{p_\theta})$, using the bilinear property of $\langle\cdot,\cdot\rangle$, it is easy to obtain
\begin{multline*}
\big<f+g,\bm{N}_{\varepsilon}^{-1}(f+g)\big> =(\varepsilon + 1)\int_{[t_{0},t)^{c}} (f_\theta(s) + g_\theta(s))^2\,\mathrm{d}s + \varepsilon\int_{t_{0}}^t(f_\theta(s) + g_\theta(s))^2\,\mathrm{d}s \\
 -\mathrm{i}\int_{t_0}^t (f_{p_{\theta}}(s)+g_{p_{\theta}}(s))^2\,\mathrm{d}s+\int_{[t_{0},t)^{c}}(f_{p_{\theta}}(s)+g_{p_{\theta}}(s))^2\,\mathrm{d}s.
\end{multline*}

In particular, for $g$ given in Eq.~\eqref{eq:g}, yields
\begin{multline}
    \big<f+g,\bm{N}_{\varepsilon}^{-1}(f+g)\big>\\
    =(\varepsilon + 1)\int_{[t_{0},t)^{c}} f^2_\theta(s) \,\mathrm{d}s + \varepsilon\int_{t_{0}}^t\left(f_\theta(s) + \frac{1}{\sqrt{\hbar m_0}}\left(\frac{-\alpha}{R} - (p_\theta(t)-p_{\theta,0})\right)\right)^2\,\mathrm{d}s \\
 -\mathrm{i}\int_{t_0}^t f^2_{p_{\theta}}(s)\,\mathrm{d}s+\int_{[t_{0},t)^{c}}f^2_{p_{\theta}}(s)\,\mathrm{d}s.
    \label{f+g}
\end{multline}
Furthermore, for $\eta_t$ given in Eq.~\eqref{eq:eta}, we have
\begin{equation*}
\big<\eta_t,\bm{N}_{\varepsilon}^{-1}(f+g)\big>=\frac{\varepsilon D}{t-t_0}\int_{t_{0}}^{t}f_{\theta}(s)\,\mathrm{d}s + \varepsilon CD + \frac{D}{t-t_0}\int_{t_{0}}^{t} f_{p_\theta}(s)\,\mathrm{d}s.
\end{equation*}
Calculating further and substituting the functions $C$ and $D$ yields
\begin{multline*}
    \big<\eta_t,\bm{N}_{\varepsilon}^{-1}(f+g)\big> \\
    =\frac{\varepsilon \sqrt{\hbar m_0}}{t-t_0}\int_{t_{0}}^{t}f_{\theta}(s)\,\mathrm{d}s + \varepsilon \left(\frac{-\alpha}{R} - (p_\theta(t)-p_{\theta,0})\right) + \frac{\sqrt{\hbar m_0}}{t-t_0}\int_{t_{0}}^{t} f_{p_\theta}(s)\,\mathrm{d}s.
\end{multline*}
This implies that
\begin{eqnarray}
u & = & \mathrm{i}(p_{\theta}(t)-p_{\theta,0})+\langle\eta_t,\bm{N}_{\varepsilon}^{-1}(f+g)\rangle\\
& = & \mathrm{i}(p_{\theta}(t)-p_{\theta,0})\nonumber\\
&   & +\frac{\varepsilon \sqrt{\hbar m_0}}{t-t_0}\int_{t_{0}}^{t}f_{\theta}(s)\,\mathrm{d}s + \varepsilon \left(\frac{-\alpha}{R} - (p_\theta(t)-p_{\theta,0})\right) + \frac{\sqrt{\hbar m_0}}{t-t_0}\int_{t_{0}}^{t} f_{p_\theta}(s)\,\mathrm{d}s. \nonumber
\end{eqnarray}
Hence, we have
	\begin{multline}
\exp\left(\frac{1}{2}u^\top(M_{\bm{N}_\varepsilon^{-1}}^{-1})u \right) = \exp\Biggl\{\frac{t-t_0}{2\varepsilon \hbar m_0}\Biggl[ \mathrm{i}(p_\theta(t) - p_{\theta,0}) +\Biggl( \frac{\varepsilon \sqrt{\hbar m_0}}{t-t_0}\int_{t_{0}}^{t}f_{\theta}(s)\,\mathrm{d}s \\
+ \varepsilon \left(\frac{-\alpha}{R} - (p_\theta(t)-p_{\theta,0})\right) + \frac{\sqrt{\hbar m_0}}{t-t_0}\int_{t_{0}}^{t} f_{p_\theta}(s)\,\mathrm{d}s \Biggr) \Biggr]^2 \Biggr\}\\
= \exp\Biggl\{\frac{t-t_0}{2\varepsilon \hbar m_0}\Biggl[ -(p_\theta(t) - p_{\theta,0})^2 + 2\mathrm{i}(p_\theta(t) - p_{\theta,0}) \Biggl( \frac{\varepsilon \sqrt{\hbar m_0}}{t-t_0}\int_{t_{0}}^{t}f_{\theta}(s)\,\mathrm{d}s \\
+ \varepsilon \left(\frac{-\alpha}{R} - (p_\theta(t)-p_{\theta,0})\right) + \frac{\sqrt{\hbar m_0}}{t-t_0}\int_{t_{0}}^{t} f_{p_\theta}(s)\,\mathrm{d}s  \Biggr) + \Biggl( \frac{\varepsilon \sqrt{\hbar m_0}}{t-t_0}\int_{t_{0}}^{t}f_{\theta}(s)\,\mathrm{d}s \\
+ \varepsilon \left(\frac{-\alpha}{R} - (p_\theta(t)-p_{\theta,0})\right) + \frac{\sqrt{\hbar m_0}}{t-t_0}\int_{t_{0}}^{t} f_{p_\theta}(s)\,\mathrm{d}s  \Biggr)^2 \Biggr] \Biggr\}.
     \label{exp-of-u}
	\end{multline}
Therefore, we have for $f = 0$ the $T$-transform of the integrand, Eq.~\eqref{I}
\begin{multline}
    TI_{V_{\mathrm{AB}},\varepsilon}(0) \\
    = \sqrt{\frac{t-t_0}{2\pi \varepsilon \hbar m_0}} \exp \Biggl\{ -\frac{\mathrm{i}p_{\theta,0}}{2\hbar m_0}  \Biggl[ \left(p_\theta' + \frac{2\alpha}{R}\right)(t-t_0)+  \left( p_\theta' - p_{\theta,0}\right)(t-t_0 + 2a) \Biggr] \Biggr\}\\
\times \exp\Biggl\{\frac{t-t_0}{2\varepsilon \hbar m_0}\Biggl[ -(p_\theta(t) - p_{\theta,0})^2 + 2\mathrm{i}\varepsilon(p_\theta(t) - p_{\theta,0})\left(\frac{-\alpha}{R} - (p_\theta(t)-p_{\theta,0})\right) \Biggr] \Biggr\}
\end{multline}
where we have chosen the normalization $\mathrm{N}=\sqrt{\det (Id + K)}$.  ``Switching off'' the perturbation, that is, taking the limit $\varepsilon \rightarrow 0$, we have obtained a propagator for the bound state AB effect in 2D in the polar conjugate momentum space
\begin{multline}
    TI_{V_{\mathrm{AB}}}(0)  \\
    = \delta \big( p_\theta' - p_{\theta,0} \big) \exp \Biggl\{ -\frac{\mathrm{i}p_{\theta,0}}{2\hbar m_0}  \Biggl[ \left(p_\theta' + \frac{2\alpha}{R}\right)(t-t_0)+  \left( p_\theta' - p_{\theta,0}\right)(t-t_0 + 2a) \Biggr] \Biggr\}\\
= \delta \big( p_\theta' - p_{\theta,0} \big) \exp \Biggl[ -\frac{\mathrm{i}p_{\theta,0}}{2\hbar m_0}\left(p_\theta' + \frac{2\alpha}{R}\right)(t-t_0) \Biggr].
    \label{TI}
\end{multline}
Similar to that in Ref.~\cite{Bock2013}, the delta function here serves to conserve the magnitude of the conjugate momentum.

\begin{pro}  
For all $p_\theta', p_{\theta,0}, t > t_0 \geq 0$, the propagator 
\begin{equation}
    K(p_\theta', t | p_{\theta,0}, t_0) := TI_{V_{\mathrm{AB}}}(0) = \delta \left( p_\theta' -p_{\theta,0} \right) \exp \Biggl[ -\frac{\mathrm{i}p_{\theta,0}}{2m_0\hbar} \Big( p_{\theta,0} - \frac{e\phi}{\pi \hbar cR} \Big) (t-t_0) \Biggr],
\end{equation}  
    solves the Schr\"{o}dinger equation
\begin{multline}
    \mathrm{i}\hbar \frac{\partial}{\partial t} K(p_\theta', t | p_{\theta,0}, t_0) \\
    = \frac{(p_\theta')^2}{2m_0} K(p_\theta', t | p_{\theta,0}, t_0) + \int_{-\infty}^\infty \mathrm{d}p_{\theta,1} W(p_\theta' - p_{\theta,1})K(p_{\theta,1}, t | p_{\theta,0}, t_0),
\end{multline}
where
\begin{equation}
    W(p_\theta' - p_{\theta,1}) = \frac{1}{2\pi} \int_{0}^{2\pi} \mathrm{d}\theta \, \mathrm{e}^{-\mathrm{i}(p_\theta' - p_{\theta,1})\theta/\hbar} V(\dot{\theta}(s)),
\end{equation}
and
\begin{equation}
    V(\dot{\theta}(s)) = -\frac{e\phi}{2\pi \hbar c} \dot{\theta}(s),
\end{equation}
with the initial condition 
\begin{equation}
		\lim_{t\searrow t_0} K(p_\theta', t | p_{\theta,0}, t_0) = \delta(p_\theta'-p_{\theta,0}).
	\end{equation}
\end{pro}

\begin{proof}
Note that, using $\alpha := (-e\phi)/(2\pi \hbar c)$ we can write $V(\dot{\theta}(s))$ as
    \begin{equation}
        V(\dot{\theta}(s)) = -\frac{e\phi}{2\pi \hbar c} \dot{\theta}(s) = -\frac{e\phi}{2\pi \hbar c} \frac{p_\theta(s)}{ m_0 R} = \alpha \frac{p_\theta(s)}{m_0 R}.
    \end{equation}
    Imposing that the conjugate momentum has always been conserved, that is, in this case, $p_\theta' = p_{\theta,1}$ and $p_\theta = p_\theta'$, we have the following
    \begin{equation}
        W(p_\theta' - p_{\theta,1}) = \alpha \frac{p_{\theta,1}}{m_0 R} \frac{1}{2\pi} \int_{0}^{2\pi} \mathrm{d}\theta \, \mathrm{e}^{-\mathrm{i}(0)\theta/\hbar} = \alpha \frac{p_{\theta,1}}{m_0 R}.
    \end{equation}
    Now,
    \begin{multline}
        \int_{-\infty}^\infty \mathrm{d}p_{\theta,1} W(p_\theta' - p_{\theta,1})K(p_{\theta,1}, t | p_{\theta,0}, t_0) \\
        = \frac{\alpha}{m_0 R} \int_{-\infty}^\infty \mathrm{d}p_{\theta,1} \, (p_{\theta,1}) \delta \left( p_{\theta,1} -p_{\theta,0} \right) \exp \Biggl[ -\frac{\mathrm{i}p_{\theta,0}}{2m_0 \hbar} \Big( p_{\theta,0} + \frac{2\alpha}{R} \Big) (t-t_0) \Biggr] \\
        = \frac{\alpha}{m_0 R} \exp \Biggl[ -\frac{\mathrm{i}p_{\theta,0}}{2m_0 \hbar} \Big( p_{\theta,0} + \frac{2\alpha}{ R} \Big) (t-t_0) \Biggr] \int_{-\infty}^\infty \mathrm{d}p_{\theta,1} \, (p_{\theta,1}) \delta \left( p_{\theta,1} -p_{\theta,0} \right) \\
        = \left(\frac{\alpha p_{\theta,0}}{m_0 R}\right) \exp \Biggl[ -\frac{\mathrm{i}p_{\theta,0}}{2m_0 \hbar} \Big( p_{\theta,0} + \frac{2\alpha}{ R} \Big) (t-t_0) \Biggr].
    \end{multline}
Therefore,
\begin{multline}
    \mathrm{i}\hbar\frac{\partial}{\partial t} K(p_\theta', t | p_{\theta,0}, t_0) \\
    = \frac{(p_\theta')^2}{2m_0} K(p_\theta', t | p_{\theta,0}, t_0) + \int_{-\infty}^\infty \mathrm{d}p_{\theta,1} W(p_\theta' - p_{\theta,1})K(p_{\theta,1}, t | p_{\theta,0}, t_0) \\
    \implies \frac{p_{\theta,0}}{2m_0} \Big( p_{\theta,0} + \frac{2\alpha}{R} \Big) \delta(p_\theta' - p_{\theta,0} ) \exp \Biggl[ -\frac{\mathrm{i}p_{\theta,0}}{2m_0 \hbar} \Big( p_{\theta,0} + \frac{2\alpha}{R} \Big) (t-t_0) \Biggr] \\
    = \frac{(p_\theta')^2}{2m_0} \delta(p_\theta' - p_{\theta,0} )\exp \Biggl[ -\frac{\mathrm{i}p_{\theta,0}}{2m_0 \hbar} \Big( p_{\theta,0} + \frac{2\alpha}{R} \Big) (t-t_0) \Biggr] \\ 
    +\left(\frac{\alpha p_{\theta,0}}{m_0 R} \right) \exp \Biggl[ -\frac{\mathrm{i}p_{\theta,0}}{2m_0 \hbar} \Big( p_{\theta,0} + \frac{2\alpha}{R} \Big) (t-t_0) \Biggr] \\
    \implies \frac{p_{\theta,0}}{2m_0} \Big( p_{\theta,0} + \frac{2\alpha}{R} \Big)   = \frac{(p_{\theta,0})^2}{2m_0} +\frac{\alpha p_{\theta,0}}{m_0 R}. 
\end{multline}
\vspace{0.3mm}
\end{proof}

In the absence of the magnetic flux, that is, $\alpha = 0$, the propagator, Eq.~\eqref{TI}, reduces to that for the case of a particle on a circle going less than a single winding around the circle (either cw or ccw), given by
\begin{equation}
    TI_{0}(0) := TI_{V_{\mathrm{AB}}=0}(0)     = \delta \left( p_\theta' -p_{\theta,0} \right) \exp \Biggl[ -\frac{\mathrm{i}p^2_{\theta,0}}{2m_0 \hbar} (t-t_0) \Biggr],
    \label{freeparticleinacircle}
\end{equation}
which also solves the corresponding Schr\"{o}dinger equation.

\subsection{With Winding Case} \label{case2}
Let us now consider the case where the charged particle is allowed to have the possibility of at least one complete winding in either \textit{cw} or \textit{ccw} direction (or a combination of both).  Furthermore, we impose the assumption that the winding number has a direct effect not only on the angular position of the charged particle but also on its conjugate momentum. To account for this, we modify the expression of the delta function in Eq.~\eqref{IAB} to have the form
\begin{equation}
    \sum_{l = -\infty}^\infty \delta \left( p_\theta(t) - p_\theta' + 2\pi m_0 R \frac{l-l_0}{t-t_0} \right),
\end{equation}
where $l$ is the integer quantum number associated with the winding (see, e.g., Refs.~\cite{edwards1967statistical, gerry1979feynman, BB02}), which takes negative values when cw (of $|l+1|$ times around the circle) and positive when ccw (of $l$ times around the circle), and $l_0$ is the value of $l$ at time $t_0$.  Thus, the Feynman integrand for this case has the form
\begin{eqnarray}
    I_{V_{\mathrm{AB}}} &:=& \mathrm{N} \exp \Biggl[ {\frac{\mathrm{i}}{\hbar} S(\omega)} \Biggr] \exp \Biggl[ {\frac{1}{2}\int_{t_0}^t (\omega^2_\theta(s) + \omega^2_p(s))\, \mathrm{d}s} \Biggr] \nonumber\\
    && \times \sum_{l = -\infty}^\infty \delta \left( p_\theta(t) - p_\theta' + 2\pi m_0 R \frac{l-l_0}{t-t_0} \right).
    \label{IAB2}
\end{eqnarray}
Employing the Poisson sum formula (see, e.g., Refs.~\cite{guinand1941poisson, bellman2013brief}), which, in our case (setting $l_0 = 0$, where we begin with zero winding at $t_0$) is
\begin{equation}
    \sum_{l = -\infty}^\infty \delta \left( p_\theta(t) - p_\theta' + 2\pi m_0 R \frac{l}{t-t_0} \right) = \frac{t-t_0}{2\pi m_0 R} \sum_{l = -\infty}^\infty \exp \left[ \mathrm{i} l \frac{\left( p_\theta(t) - p_\theta' \right)(t-t_0)}{m_0 R} \right],
\end{equation}
we have (compare with Eq.~\eqref{I})
\begin{multline}
    I_{V_{\mathrm{AB}}} 
    = \mathrm{N} \frac{t-t_0}{2\pi m_0 R}\exp \Biggl\{ -\frac{\mathrm{i}p_{\theta,0}}{2\hbar m_0} \left[ \left(p_\theta(t) + \frac{2\alpha}{R}\right)(t-t_0) +  \left( p_\theta(t) - p_{\theta,0}\right)(t-t_0 + 2a) \right] \Biggr\} \\
    \times \exp \Biggl( -\frac{1}{2} \langle (\omega_\theta, \omega_p), K (\omega_\theta, \omega_p) \rangle \Biggr) \exp \Biggl( \mathrm{i} \langle (\omega_\theta, \omega_p), g \rangle \Biggr) \\  
    \times  \sum_{l = -\infty}^\infty \exp \left[ \mathrm{i} l \frac{t-t_0}{m_0 R}\left( p_{\theta,0} - p_\theta' \right) \right] \exp \bigg( \mathrm{i}  \left\langle (\omega_\theta, \omega_p), g + k_l \right\rangle  \bigg),
\end{multline}
where $k_l :=  l \frac{(t-t_0)}{m_0R} \eta_t $.

Taking the \textit{T}-transform of $I_{V_{\mathrm{AB}}}$ at the point $f = (f_\theta, f_{p_\theta})$, yields
\begin{multline}
    TI_{V_{\mathrm{AB}}}(f) \\
    = \frac{t-t_0}{2\pi m_0 R} \exp \Biggl\{ -\frac{\mathrm{i}p_{\theta,0}}{2\hbar m_0} \left[ \left(p_\theta' + \frac{2\alpha}{R}\right)(t-t_0) +  \left( p_\theta' - p_{\theta,0}\right)(t-t_0 + 2a) \right] \Biggr\} \\
    \times \sum_{l = -\infty}^\infty \exp \left[ \mathrm{i} l \frac{t-t_0}{m_0 R}\left( p_{\theta,0} - p_\theta' \right) \right] \int_{S_2'} \mathrm{N} \exp \Biggl( -\frac{1}{2} \langle (\omega_\theta, \omega_p), K (\omega_\theta, \omega_p) \rangle \Biggr)\\
    \times \exp \bigg( \mathrm{i}  \left\langle (\omega_\theta, \omega_p), g + k_l \right\rangle  \bigg) \exp \bigg( \mathrm{i} \left\langle (\omega_\theta, \omega_p), f \right\rangle \bigg) \mathrm{d}\mu(\omega) \\
    = \frac{t-t_0}{2\pi m_0 R} \exp \Biggl\{ -\frac{\mathrm{i}p_{\theta,0}}{2\hbar m_0} \left[ \left(p_\theta' + \frac{2\alpha}{R}\right)(t-t_0) +  \left( p_\theta' - p_{\theta,0}\right)(t-t_0 + 2a) \right] \Biggr\} \\
    \times \sum_{l = -\infty}^\infty \exp \left[ \mathrm{i} l \frac{t-t_0}{m_0 R}\left( p_{\theta,0} - p_\theta' \right) \right] \\
    \times T\left( \mathrm{N} \exp \Biggl( -\frac{1}{2} \langle (\omega_\theta, \omega_p), K (\omega_\theta, \omega_p) \rangle \Biggr)\right)(f + g + k_l).
\end{multline}
In particular, following Lemma~\ref{lemma}, we have
\begin{multline}
    T\left( \mathrm{N} \exp \Biggl( -\frac{1}{2} \langle \cdot, K \cdot \rangle \Biggr)\right)(f + g + k_l) \\
    = \left(\frac{2\pi m_0 R}{t-t_0}\sqrt{\det({\mathrm{Id}}+{K})}\right)\frac{1}{\sqrt{\det({\mathrm{Id}}+{K})}} \\
    \times\exp \left(-\frac{1}{2}\left\langle {f+g + k_l}, ({\mathrm{Id}}+{K})^{-1}{(f+g + k_l)}\right\rangle \right)\\
    = \left(\frac{2\pi m_0 R}{t-t_0}\right)\exp \left(-\frac{1}{2}\left\langle {f+g + k_l}, ({\mathrm{Id}}+{K})^{-1}{(f+g + k_l)}\right\rangle \right),
\end{multline}
where we have chosen the normalization $\mathrm{N}=\left(({2\pi m_0 R}/{(t-t_0)})\sqrt{\det({\mathrm{Id}}+{K})}\right)$ in such a way to also ``eliminate'' $(t-t_0)/2\pi m_0 R$ and the divergent factor $1/{\sqrt{\det({\mathrm{Id}}+{K})}}$.

With
\begin{multline}
   \left\langle {f+g + k_l}, ({\mathrm{Id}}+{K})^{-1}{(f+g + k_l)}\right\rangle 
    \\
    =  \int_{[t_{0},t)^{c}} \left( f^2_\theta(s) + f^2_{p_{\theta}}(s)\right)\mathrm{d}s  -\mathrm{i}\int_{t_0}^t f^2_{p_{\theta}}(s)\,\mathrm{d}s,
\end{multline}
we have
\begin{multline}
    TI_{V_{\mathrm{AB}}}(f) = \exp \Biggl\{ -\frac{\mathrm{i}p_{\theta,0}}{2\hbar m_0} \left[ \left(p_\theta' + \frac{2\alpha}{R}\right)(t-t_0) +  \left( p_\theta' - p_{\theta,0}\right)(t-t_0 + 2a) \right] \Biggr\} \\
    \times \exp \Biggl\{ -\frac{1}{2} \Biggl[ \int_{[t_{0},t)^{c}} \left( f^2_\theta(s) + f^2_{p_{\theta}}(s)\right)\,\mathrm{d}s  -\mathrm{i}\int_{t_0}^t f^2_{p_{\theta}}(s)\,\mathrm{d}s \Biggr] \Biggr\} \\
    \times \sum_{l = -\infty}^\infty \exp \left[ \mathrm{i} l \frac{t-t_0}{m_0 R}\left( p_{\theta,0} - p_\theta' \right) \right]. 
 \label{withwindingTI(f)}
\end{multline}
Setting $f=0$, we obtain the quantum propagator in the polar conjugate momentum space for the bound state AB effect for $t > t_0$
\begin{multline}
    TI_{V_{\mathrm{AB}}}(0) = 
    \exp \Biggl\{ -\frac{\mathrm{i}p_{\theta,0}}{2\hbar m_0} \left[ \left(p_\theta' + \frac{2\alpha}{R}\right)(t-t_0) +  \left( p_\theta' - p_{\theta,0}\right)(t-t_0 + 2a) \right] \Biggr\} \\
    \times \sum_{l = -\infty}^\infty \exp \left[ \mathrm{i} l \frac{t-t_0}{m_0 R}\left( p_{\theta,0} - p_\theta' \right) \right].
    \label{TI2}
\end{multline}
We note here that Eq.~\eqref{TI2} suggests that as the conservation of the magnitude of the polar conjugate momentum is imposed, that is, in the limit $p_\theta' \rightarrow p_{\theta,0}$, the propagator reduces to that in Eq.~\eqref{TI}.

\begin{pro}  
If $p_\theta' = p_{\theta,0}$, $t > t_0 \geq 0$, then the propagator 
\begin{eqnarray}
    K(p_\theta', t | p_{\theta,0}, t_0) &:=& TI_{V_{\mathrm{AB}}}(0) \nonumber\\
    &=& \exp \Biggl\{ -\frac{\mathrm{i}p_{\theta,0}}{2\hbar m_0} \left[ \left(p_\theta' + \frac{2\alpha}{R}\right)(t-t_0) +  \left( p_\theta' - p_{\theta,0}\right)(t-t_0 + 2a) \right] \Biggr\}  \nonumber\\
    && \times \sum_{l = -\infty}^\infty \exp \left[ \mathrm{i} l \frac{t-t_0}{m_0 R}\left( p_{\theta,0} - p_\theta'  \right) \right],
    \label{prop2}
\end{eqnarray}  
    solves the Schr\"{o}dinger equation
\begin{multline}
    \mathrm{i}\hbar \frac{\partial}{\partial t} K(p_\theta', t | p_{\theta,0}, t_0) \\
    = \frac{(p_\theta')^2}{2m_0} K(p_\theta', t | p_{\theta,0}, t_0) + \int_{-\infty}^\infty \mathrm{d}p_{\theta,1} W(p_\theta' - p_{\theta,1})K(p_{\theta,1}, t | p_{\theta,0}, t_0),
    \label{schro2}
\end{multline}
where
\begin{equation}
    W(p_\theta' - p_{\theta,1}) = \frac{1}{2\pi } \int_{0}^{2\pi} \mathrm{d}\theta \, \mathrm{e}^{-\mathrm{i}(p_\theta' - p_{\theta,1})\theta/\hbar} V(\dot{\theta}(s)),
\end{equation}
and
\begin{equation}
    V(\dot{\theta}(s)) = -\frac{e\phi}{2\pi \hbar c} \dot{\theta}(s),
\end{equation}
with the initial condition 
\begin{equation}
		\lim_{t\searrow t_0} K(p_\theta', t | p_{\theta,0}, t_0) = \delta(p_\theta'-p_{\theta,0}).
	\end{equation}
\end{pro}

\begin{proof}
This can be verified explicitly by trying to solve directly the Schr\"{o}dinger equation using the propagator Eq.~\eqref{prop2}.  The equality of the Schr\"{o}dinger equation can be shown to hold only when preserving the magnitude of the conjugate momentum.

We begin with the lhs of Eq.~\eqref{schro2}:
    \begin{multline}
    \mathrm{i}\hbar \frac{\partial}{\partial t} K(p_\theta', t | p_{\theta,0}, t_0) \\
    = \mathrm{i}\hbar 
    \exp \Biggl\{ -\frac{\mathrm{i}p_{\theta,0}}{2\hbar m_0} \left[ \left(p_\theta' + \frac{2\alpha}{R}\right)(t-t_0) +  \left( p_\theta' - p_{\theta,0}\right)(t-t_0 + 2a) \right] \Biggr\} \\
    \times\sum_{l=-\infty}^\infty  \exp \left[ \mathrm{i} l \frac{t-t_0}{m_0 R}\left( p_{\theta,0} - p_\theta' \right) \right] \\
    \times  
    \Biggl[ -\frac{\mathrm{i}p_{\theta,0}}{2m_0\hbar} \Big( p_\theta'   - \frac{e\phi}{\pi \hbar cR} + p_\theta' - p_{\theta,0} \Big) + \left(\mathrm{i} l \frac{1}{m_0 R}  \right)\left( p_{\theta,0} - p_\theta' \right) \Biggr].
\end{multline}
Now, note that the term $\left(\mathrm{i} l \frac{1}{m_0 R} +\frac{\mathrm{i}p_{\theta,0}}{2m_0\hbar} \right)\left( p_{\theta,0} - p_\theta' \right)$ has no counterpart in the rhs of the Schr\"{o}dinger equation since the only expression
\begin{equation}
    W(p_\theta' - p_{\theta,1}) = -\frac{e\phi}{2\pi \hbar c}\frac{1}{2\pi} \int_{0}^{2\pi} \mathrm{d}\theta \, \mathrm{e}^{-\mathrm{i}(p_\theta' - p_{\theta,1})\theta/\hbar}  \dot{\theta}
\end{equation}
that could contain this term has the potential term $\alpha := \frac{-e\phi}{2\pi \hbar c}$.  Thus, the only way around this is when $p_\theta' =  p_{\theta,0}$, so as to make $\left(\mathrm{i} l \frac{1}{m_0 R} +\frac{\mathrm{i}p_{\theta,0}}{2m_0\hbar} \right)\left( p_{\theta,0} - p_\theta' \right) = 0$ and the Schr\"{o}dinger equation is satisfied.

\end{proof}

Having examined the above two cases, namely, ``Incomplete Winding Case'' and ``With Winding Case'' (see Subsections~\ref{case1} and \ref{case2}, respectively), we are therefore led to the following main result.
\begin{thm}[Particle in a Bound State AB Potential Momentum Space Propagator]
For a charged particle undergoing a bound state AB effect as described and depicted in Fig.~\ref{solenoid1}, the momentum- (specifically the polar-conjugate-momentum-) space propagator is given by
    \begin{multline}
        K(p_\theta', t | p_{\theta,0}, t_0) =    \delta \left( p_\theta' -p_{\theta,0} \right) \exp \Biggl[ -\frac{\mathrm{i}p_{\theta,0}}{2m_0\hbar} \Big( p_{\theta,0} - \frac{e\phi}{\pi \hbar cR} \Big) (t-t_0) \Biggr].
    \end{multline}
    This further implies that the magnitude of the conjugate momentum of the particle is conserved through the time interval $[t_0, t]$ with $t > t_0 \geq 0$, regardless of the winding number $l$ taken by the motion of the particle.
\end{thm}

\begin{rem}[Particle on a Circle Momentum Space Propagator]
    The main result also implies that, for the case where the magnetic flux $\phi=0$, which is the case of a particle on a circle, the momentum- (specifically the polar-conjugate-momentum-) space propagator is given by
    \begin{equation}
        K(p_\theta', t | p_{\theta,0}, t_0) =    \delta \left( p_\theta' -p_{\theta,0} \right) \exp \Biggl[ -\frac{\mathrm{i}p_{\theta,0}^2}{2m_0\hbar} (t-t_0) \Biggr],
    \end{equation}
    which also obeys the conservation of momentum regardless of the winding number similar to that in the main result above.
\end{rem}

\section{Alternative WNA Approach in Achieving the Bound State AB Effect}\label{sec:Alternative}
\label{app:alternative}

Here we outline another WNA approach for the bound state AB effect based on a perturbation series of the Feynman integrand for the particle on a circle.

\subsection{Particle on a Circle Perturbed by a Potential that is Growing Exponentially in Angular Velocity}\label{subsec: perturbed}

\begin{defn}[cf.~Ref.~\cite{KSW99}]
We define a potential $V: \mathbb{R} \longrightarrow \mathbb{R}$ by
\begin{equation}
        V(x) = \int_{\mathbb{R}} \mathrm{e}^{\beta x} \, \mathrm{d}m(\beta) \qquad x \in \mathbb{R},
    \end{equation} 
with $m$ being a complex measure on the Borel sets on $\mathbb{R}$ fulfilling the condition
\begin{equation}
        \int_{\mathbb{R}} \mathrm{e}^{C|\beta|} \, \mathrm{d}|m|(\beta) < \infty, \qquad \forall  C > 0.
        \label{finitecondition} 
    \end{equation}
For each $s \in [t_0, t]$, setting $x:= \mathrm{d}\theta/\mathrm{d}s =: \dot{\theta}(s)$, we have adapted a version of the exponentially growing potential, $V(\dot{\theta}(s))$, which we call a class of potentials that are exponentially growing in angular velocity.
    
\end{defn}

We first set up the Feynman integrand as a perturbation series of the form
\begin{multline}
    I_V = I_{0} \exp \Biggl( -\frac{\mathrm{i}}{\hbar} \int_{t_0}^t V(\dot{\theta}(s)) \, \mathrm{d}s \Biggr) \\
    = I_{0} \sum_{n=0}^\infty \frac{(-\mathrm{i}/\hbar)^n}{n!} \int_{[t_0,t)^n} \mathrm{d}^n s \int_{\mathbb{R}^n} \prod_{j=1}^n \mathrm{d}m(\beta_j) \, \mathrm{e}^{\sum_{j=1}^n \beta_j \dot{\theta}(s_j)} \\
    =  \sum_{n=0}^\infty \frac{(-\mathrm{i/\hbar})^n}{n!} \int_{[t_0,t)^n} \mathrm{d}^n s \int_{\mathbb{R}^n} \prod_{j=1}^n \mathrm{d}m(\beta_j) \, I_{0} \mathrm{e}^{\sum_{j=1}^n \beta_j \dot{\theta}(s_j)}.
\end{multline}
Then, we need to show that $I_V$ is a white noise object before we can proceed to solve the propagator.  

Consider the functional $\Psi :=  I_{0} \mathrm{e}^{\sum_{j=1}^n \beta_j \dot{\theta}(s_j)}$.  As an adaptation of Refs.~\cite{bock2014hamiltonian, KSW99}, we can have our version for the polar coordinates starting with Eq.~\eqref{theta} to have
\begin{equation}
p_\theta(s_j) = p_{\theta,0} + D\omega_\theta(s_j) = p_{\theta,0} + \big< \omega_\theta, D\delta_{s_j}\big>,
\end{equation}
or
\begin{multline}
p_\theta(t) = p_{\theta}(s_j) + D(\omega_\theta(t) - \omega_\theta(s_j)) = p_{\theta}(s_j) + \big< \omega_\theta, D(\delta_{t}-\delta_{s_j}) \big> \\
\implies p_{\theta}(s_j) = p_\theta(t) -  \big< \omega_\theta, D(\delta_{t}-\delta_{s_j}) \big>.
\end{multline}
With Theorem.~\ref{Th1}, it is enough to characterize $\Psi$ with its \textit{T}-transform at the point $\varphi = (\varphi_\theta, \varphi_{p_\theta})$:
\begin{multline}
    T\Psi(\varphi) = \int_{S_2'} \Psi(\omega) \exp \left(\mathrm{i}\left\langle \omega, \varphi \right\rangle \right) \mathrm{d}\mu (\omega) \\
    = \int_{S_2'} I_{0}(\omega) \exp \left[ \sum_{j=1}^n  \beta_j \Biggl(\frac{p_{\theta}(t)}{m_0R} - \big< \omega_\theta, E (\delta_{t}-\delta_{s_j}) \big> \Biggr)\right] \exp \left(\mathrm{i}\left\langle \omega, \varphi \right\rangle \right) \mathrm{d}\mu (\omega) \\
    = \exp \left( \frac{p_{\theta}(t)}{m_0R}\sum_{j=1}^n \beta_j \right) \\
    \times\int_{S_2'} I_{0}(\omega) \exp \left( \mathrm{i}   \left\langle (\omega_{\theta}, \omega_p), \mathrm{i}\left(E\sum_{j=1}^n \beta_j (\delta_{t}-\delta_{s_j}),\, 0\right) \right\rangle  \right)  \exp \left(\mathrm{i}\left\langle \omega, \varphi \right\rangle \right) \mathrm{d}\mu (\omega) \\
    = \exp \left( \frac{p_{\theta}'}{m_0R}\sum_{j=1}^n \beta_j \right) 
     \int_{S_2'} I_{0}(\omega) \exp \left( \mathrm{i}   \left\langle \omega, \mathrm{i}\xi_t \right\rangle  \right) \exp \left(\mathrm{i}\left\langle \omega, \varphi \right\rangle \right) \mathrm{d}\mu (\omega) \\
    = \exp \left( \frac{p_{\theta}'}{m_0R}\sum_{j=1}^n \beta_j \right)  
     TI_{0}(\varphi + \mathrm{i}\xi_t)
    \label{appendix: T-transform}
\end{multline}
where $\xi_t := \left(E\sum_{j=1}^n \beta_j(\delta_{t}-\delta_{s_j}) ,\, 0 \right)$.

We note here that since both of the cases considered above (see Subsections~\ref{case1} and \ref{case2}) produce results that have the conservation of the magnitude of the conjugate momentum, we may as well choose the calculations of the \textit{T}-transform of the particle on a circle integrand that from either of the two (i.e., when $\alpha = 0$), which we expect to produce the same results. 

With Eqs.~\eqref{M}, \eqref{f+g}, and \eqref{exp-of-u}, we may explicitly calculate for Eq.~\eqref{appendix: T-transform} with a small perturbation $\varepsilon > 0$ using Lemma~\ref{lemma} (let $\xi_{t,\theta}=E\sum_{j=1}^n \beta_j(\delta_{t}-\delta_{s_j})$):
\begin{multline}
    T\Psi_\varepsilon(\varphi) = \exp \left( \frac{p_{\theta}'}{m_0R}\sum_{j=1}^n \beta_j \right)  \frac{\sqrt{t-t_0}}{\sqrt{2\pi \hbar m_0 \varepsilon}} \\
    \times \exp \Biggl\{ -\frac{\mathrm{i}p_{\theta,0}}{2\hbar m_0}  \Biggl[ p_\theta'(t-t_0)+  \left( p_\theta' - p_{\theta,0}\right)(t-t_0 + 2a) \Biggr] \Biggr\}\\
    \times \exp \Biggl\{ -\frac{1}{2} \Biggl[ (\varepsilon + 1)\int_{[t_{0},t)^{c}} (\varphi_\theta(s) + \xi_{t,\theta}(s))^2 \,\mathrm{d}s \\
    + \varepsilon\int_{t_{0}}^t \left(\varphi_\theta(s) + \xi_{t,\theta}(s) - \frac{1}{\sqrt{\hbar m_0}} (p_\theta(t)-p_{\theta,0})\right)^2\,\mathrm{d}s \\
 -\mathrm{i}\int_{t_0}^t \varphi^2_{p_{\theta}}(s)\,\mathrm{d}s+\int_{[t_{0},t)^{c}}\varphi^2_{p_{\theta}}(s)\,\mathrm{d}s
 \Biggr] + \frac{t-t_0}{2\varepsilon \hbar m_0}\Biggl[ -(p_\theta' - p_{\theta,0})^2 \\
 + 2\mathrm{i}(p_\theta' - p_{\theta,0}) \biggl[\varepsilon \biggl( \frac{\sqrt{\hbar m_0}}{t-t_0}  \int_{t_0}^t (\varphi_\theta(s) + \xi_{t,\theta}(s))\,\mathrm{d}s - (p_\theta'-p_{\theta,0}) \biggr) \\
 + \frac{\sqrt{\hbar m_0}}{t-t_0} \int_{t_0}^t \varphi_{p_\theta}(s)\,\mathrm{d}s \biggr] + \biggl[\varepsilon \biggl( \frac{\sqrt{\hbar m_0}}{t-t_0}  \int_{t_0}^t  (\varphi_\theta(s) + \xi_{t,\theta}(s))\,\mathrm{d}s - (p_\theta'-p_{\theta,0}) \biggr) \\
 + \frac{\sqrt{\hbar m_0}}{t-t_0} \int_{t_0}^t \varphi_{p_\theta}(s)\,\mathrm{d}s \biggr]^2 \Biggr]  \Biggr\} \\
= TI_{0,\varepsilon}(\varphi) G_\varepsilon(p_\theta',p_{\theta,0},\xi_{t,\theta},t,t_0),
    \label{appendix: T-transform2}
\end{multline}
where
\begin{multline}
     G_\varepsilon(p_\theta',p_{\theta,0},\xi_{t,\theta},t,t_0) = \exp \left( \frac{p_{\theta}'}{m_0R}\sum_{j=1}^n \beta_j \right) \\
     \times \exp \Biggl\{ -\frac{1}{2} \Biggl[ (\varepsilon + 1)\int_{[t_{0},t)^{c}} (2\varphi_\theta(s) \xi_{t,\theta}(s) + \xi^2_{t,\theta}(s)) \,\mathrm{d}s + \varepsilon\int_{t_{0}}^t\biggl( 2\xi_{t,\theta}(s) \Big( \varphi_\theta(s) \\
    - \frac{1}{\sqrt{\hbar m_0}} (p_\theta(t)-p_{\theta,0})\Big) + \Big(\xi_{t,\theta}(s) \Big)^2 \biggr)\,\mathrm{d}s \Biggr] \\
    + \frac{t-t_0}{2\varepsilon \hbar m_0}\Biggl[  2\mathrm{i}(p_\theta' - p_{\theta,0}) \varepsilon  \frac{\sqrt{\hbar m_0}}{t-t_0}  \int_{t_0}^t  \xi_{t,\theta}(s)\,\mathrm{d}s   + 2\varepsilon^2\frac{\sqrt{\hbar m_0}}{t-t_0}  \int_{t_0}^t  \xi_{t,\theta}(s) \,\mathrm{d}s \\
    \times\biggl( \frac{\sqrt{\hbar m_0}}{t-t_0}  \int_{t_0}^t  \varphi_\theta(s)\,\mathrm{d}s - (p_\theta'-p_{\theta,0}) + \frac{\sqrt{\hbar m_0}}{\varepsilon (t-t_0)} \int_{t_0}^t \varphi_{p_\theta}(s)\,\mathrm{d}s \biggr) \\
 +\biggl(\varepsilon\frac{\sqrt{\hbar m_0}}{t-t_0}  \int_{t_0}^t  \xi_{t,\theta}(s)\,\mathrm{d}s \biggr)^2 \Biggr]  \Biggr\}.
\end{multline}
Taking the limit of $\varepsilon \rightarrow 0$, we have
\begin{multline}
    T\Psi(\varphi) =  TI_{0}(\varphi) G(p_\theta',p_{\theta,0},\xi_{t,\theta},t,t_0) \\
    = \exp \left( \frac{p_{\theta}'}{m_0R}\sum_{j=1}^n \beta_j \right) TI_{0}(\varphi) \exp \Biggl\{ \frac{1}{t-t_0}\int_{t_{0}}^{t}\xi_{t,\theta}(s)\,\mathrm{d}s\int_{t_{0}}^{t}\varphi_{p_\theta}(s)\,\mathrm{d}s \\
 + \frac{t-t_0}{\hbar m_0}\Biggl[ \mathrm{i}(p_\theta' - p_{\theta,0}) \frac{\sqrt{\hbar m_0}}{t-t_0}  \int_{t_0}^t  \xi_{t,\theta}(s)\,\mathrm{d}s  \Biggr]  \Biggr\} \\
 = \exp \left( \frac{p_{\theta}'}{m_0R}\sum_{j=1}^n \beta_j \right) TI_{0}(\varphi) \exp \left( \frac{1}{t-t_0}\int_{t_{0}}^{t}\xi_{t,\theta}(s)\,\mathrm{d}s\int_{t_{0}}^{t}\varphi_{p_\theta}(s)\,\mathrm{d}s \right) \\
 = \exp \left( \frac{p_{\theta}'}{m_0R}\sum_{j=1}^n \beta_j \right) TI_{0}(\varphi) \exp \left( \frac{1}{t-t_0}\sqrt{\frac{\hbar}{m_0R^2}}\int_{t_{0}}^{t}\varphi_{p_\theta}(s)\,\mathrm{d}s\sum_{j=1}^n \beta_j(t-s_j) \right) 
    \label{appendix: T-transform4}
\end{multline}
since the exponential term containing $(p_\theta' - p_{\theta,0})$ vanishes due to the delta function that serves to conserve the magnitude of the conjugate momentum contained in the $TI_{0}(\varphi)$.  As already pointed out, we would also have obtained a result similar to that in Eq.~\eqref{appendix: T-transform4} had we used the \textit{T}-transform of the integrand in Eq.~\eqref{withwindingTI(f)} (with $\alpha$ = 0).

\begin{pro}
    Let $n \in \mathbb{N}_0$, $s_j \in [t_0,t]$, and $\beta_j \in \mathbb{R}$, $1 \leq j \leq n$ be given.  Then, the product
    \begin{equation}
        \Psi_n =  I_{0} \exp \left(\sum_{j=1}^n \beta_j \dot{\theta}(s_j)  \right)
    \end{equation}
    defined for any $\varphi \in S_2$ by
    \begin{equation}
        T\Psi_n(\varphi) = \exp \left( \dot{\theta}(t)\sum_{j=1}^n \beta_j  \right) TI_{0}(\varphi + \mathrm{i}\xi_t)
    \end{equation}
    is a well-defined element in the Kondratiev distribution space $(S_2)^{-1}$.
\end{pro}
\begin{proof}
    Let 
\begin{equation}
    T\Psi_n(\varphi) = TI_{0}(\varphi) G_n(\varphi, \xi_t),
    \label{TofPsi}
\end{equation}
with
\begin{equation}
    G_n(\varphi, \xi_t) := \exp \left( \frac{p_{\theta}'}{m_0R}\sum_{j=1}^n \beta_j \right) \exp \left( \frac{1}{t-t_0}\sqrt{\frac{\hbar}{m_0R^2}}\int_{t_{0}}^{t}\varphi_{p_\theta}(s)\,\mathrm{d}s\sum_{j=1}^n \beta_j(t-s_j) \right).
\end{equation}
We note that in Eq.~\eqref{TofPsi}, since $I_{0}(\varphi)$ has already been shown to be a Hida distribution as guaranteed by Lemma~\ref{lemma}, hence we only need to show that $G_n(\varphi, \xi_t)$ satisfies the conditions of Definition~\ref{def:holmorphy} and Theorem~\ref{Th1}.  This will show that $\Psi_n$ is an element in $(S_2)^{-1}$.  As $G_n(\varphi_0 + z\varphi, \xi_t)$ entire in $z$ being obvious, we proceed to calculate the boundedness of $G_n(\varphi, \xi_t)$:
\begin{multline}
    \left| G_n(\varphi, \xi_t) \right| 
    \leq \exp \left| \frac{p_{\theta}'}{m_0 R}\sum_{j=1}^n \beta_j  \right|\exp \left| \frac{1}{t-t_0}\sqrt{\frac{\hbar}{m_0R^2}}\int_{t_{0}}^{t}\varphi_{p_\theta}(s)\,\mathrm{d}s\sum_{j=1}^n \beta_j(t-s_j) \right| \\ 
    \leq \exp \left( \frac{|p_{\theta}'|}{m_0 R}\sum_{j=1}^n |\beta_j| \right) \exp \left( \sqrt{\frac{\hbar}{m_0R^2}} \sqrt{t-t_0} \left|\varphi_{p_\theta, \Delta}\right|_0\sum_{j=1}^n |\beta_j| \right) \\
    = \prod_{j=1}^n \exp \left( C |\beta_j|  \right) \quad \quad \mathrm{with} \quad \quad C = C(t-t_0, \left|\varphi_{\Delta}\right|_0, |p_{\theta}|).
    \label{appendix: bound1}
\end{multline}
where $\Delta := t - t_0$.
\end{proof}

\begin{lem}\label{appendix: lemma}
    For every $n \in \mathbb{N}_0$, the integral
    \[
    \int_{[t_0,t)^n} \mathrm{d}^n s \int_{\mathbb{R}^n} \prod_{j=1}^n \mathrm{d}m(\beta_j) \, I_{0} \mathrm{e}^{\sum_{j=1}^n \beta_j \dot{\theta}(s_j)}
    \]
    defined for any $\varphi \in S_2$ by
    \begin{eqnarray}
        &&\int_{[t_0,t)^n} \mathrm{d}^n s \int_{\mathbb{R}^n} \prod_{j=1}^n \mathrm{d}m(\beta_j) \, T\left( I_{0} \mathrm{e}^{\sum_{j=1}^n \beta_j \dot{\theta}(s_j)}\right)(\varphi) \nonumber\\
        &&= TI_{0}(\varphi)\int_{[t_0,t)^n} \mathrm{d}^n s \int_{\mathbb{R}^n} \prod_{j=1}^n \mathrm{d}m(\beta_j) \, G_n(\varphi, \xi_t)
    \end{eqnarray}
\end{lem}
is a Kondratiev distribution.
\begin{proof}
    It follow from Eq.~\eqref{appendix: bound1}, the term $G_n(\varphi, \xi_t)$ is linearly bound, 
    and so $T\Psi_n(\varphi)$ is measurable for any $\varphi \in S_2$.  Thus in order to apply Corollary~\ref{Bochner} (Bochner integrability) we only need to see the absolute integrability of $G_n(\varphi, \xi_t)$ through the following bound
    \begin{eqnarray}
         && \left| \int_{[t_0,t)^n} \mathrm{d}^n s\int_{\mathbb{R}} \mathrm{d}m(\beta_1) \ldots \mathrm{d}m(\beta_n) \, G_n(\varphi, \xi_t) \right| \nonumber\\
         &&\leq \int_{[t_0,t)^n} \mathrm{d}^n s\int_{\mathbb{R}} \mathrm{d}m(\beta_1) \ldots \mathrm{d}m(\beta_n) \, \prod_{j=1}^n \exp \left( C |\beta_j|  \right) \nonumber\\
         && = \left( (t-t_0) \int_{\mathbb{R}} \exp \left( C |\beta|  \right) \mathrm{d}m(\beta) \right)^n,
         \label{bound2}
    \end{eqnarray}
which is finite since the measure satisfies the condition in Eq.~\eqref{finitecondition}.
\end{proof}

Due to Lemma~\ref{appendix: lemma} there exists an open neighborhood $U$ independent of $n$ and
\begin{equation}
    I_{V,n} := \int_{[t_0,t)^n} \mathrm{d}^n s \int_{\mathbb{R}} \prod_{j=1}^n \mathrm{d}m(\beta_j) \, \Psi_n \in (S_2)^{-1},\; \forall n \in \mathbb{N}_0
\end{equation}
is holomorphic on $U$, fulfilling first part of Corollary~\ref{seq}.  Then, recalling that $TI_0(\varphi)$ has been proven as a Hida distribution and thus a $U$-functional, we can bound $TI_{V,n}(\varphi)$ by
\begin{multline}
    \left| TI_{V}(\varphi) \right| \leq \sum_{n=0}^\infty \frac{1}{n!} \left| TI_{V,n}(\varphi) \right|(\varphi)\\
    \leq \left| TI_{0}(\varphi) \right| \sum_{n=0}^\infty \frac{(1/\hbar)^n}{n!} \int_{[t_0,t)^n} \mathrm{d}^n s \int_{\mathbb{R}^n} \prod_{j=1}^n \mathrm{d}m(\beta_j) \, \left| G_n(\varphi, \xi_t) \right| \\
    \leq \left| TI_{0}(\varphi) \right| \sum_{n=0}^\infty \frac{(1/\hbar)^n}{n!} \left( (t-t_0) \int_{\mathbb{R}} \exp \left( C |\beta|  \right) \mathrm{d}m(\beta) \right)^n \\
    = \left| TI_{0}(\varphi) \right| \exp\left( (t-t_0) \int_{\mathbb{R}} \exp \left( C |\beta|  \right) \mathrm{d}m(\beta) \right) < \infty
    \label{appendix: bound3}
\end{multline}
for $\varphi \in U$, so that we have shown that the series converges in $(S_2)^{-1}$.

We therefore have the following
\begin{thm}
    For a potential of the form
    \begin{equation}
        V(\dot{\theta}(s)) = \int_{\mathbb{R}} \mathrm{e}^{\beta \dot{\theta}(s)} \, \mathrm{d}m(\beta) \qquad \dot{\theta}(s) \in \mathbb{R},
    \end{equation}    
    where $m$ is any complex measure satisfying
    \begin{equation}
        \int_{\mathbb{R}} \mathrm{e}^{C|b|} \, \mathrm{d}|m|(b) < \infty, \qquad \forall  C > 0,
    \end{equation}
    the Feynman integrand
    \begin{multline}
    I_V = I_{0} \exp \Biggl( -\frac{\mathrm{i}}{\hbar} \int_{t_0}^t V(\dot{\theta}(s)) \, \mathrm{d}s \Biggr) \\
    =  \sum_{n=0}^\infty \frac{(-\mathrm{i/\hbar})^n}{n!} \int_{[t_0,t)^n} \mathrm{d}^n s \int_{\mathbb{R}^n} \prod_{j=1}^n \mathrm{d}m(\beta_j) \, I_{0} \exp \left( {\sum_{j=1}^n \beta_j \dot{\theta}(s_j)} \right)
\end{multline}
    exists as a generalized white noise functional.  The series converges strongly in $(S_2)^{-1}$, and the integrals exist in the sense of Bochner integrals.
\end{thm}

\begin{rem}
    It is straightforward to generalize to time-dependent potential; see Refs.~\cite{KSW99, GKSS96}.
\end{rem}

\begin{rem}
    For all $p_\theta'$, $p_{\theta,0}$, and $0 \leq t_0 < t$, 
    \begin{multline}
        K_V(p_\theta', t | p_{\theta,0}, t_0) := TI_V(0) \\
        = TI_{0}(0) \sum_{n=0}^\infty \frac{(-\mathrm{i/\hbar})^n}{n!} \int_{[t_0,t)^n} \mathrm{d}^n s \int_{\mathbb{R}^n} \prod_{j=1}^n \mathrm{d}m(\beta_j) \,  \exp \left( \frac{p_{\theta}'}{m_0 R}\sum_{j=1}^n \beta_j  \right) \\
        = \delta \left( p_\theta' -p_{\theta,0} \right) \exp \Biggl[ -\frac{\mathrm{i}p_{\theta,0}^2}{2m_0\hbar} (t-t_0) \Biggr] \\
        \times \sum_{n=0}^\infty \frac{(-\mathrm{i/\hbar})^n}{n!} \int_{[t_0,t)^n} \mathrm{d}^n s \int_{\mathbb{R}^n} \prod_{j=1}^n \mathrm{d}m(\beta_j) \,  \exp \left( \frac{p_{\theta}'}{m_0 R}\sum_{j=1}^n \beta_j  \right) \\
        = \delta \left( p_\theta' -p_{\theta,0} \right) \exp \Biggl[ -\frac{\mathrm{i}p_{\theta,0}^2}{2m_0\hbar} (t-t_0) \Biggr] \\
        \times \sum_{n=0}^\infty \frac{(-\mathrm{i/\hbar})^n}{n!} (t-t_0)^n \left[ \int_{\mathbb{R}} \mathrm{d}m(\beta) \,  \exp \left( \frac{p_{\theta}'}{m_0 R} \beta  \right) \right]^n \\
        = \delta \left( p_\theta' -p_{\theta,0} \right) \exp \Biggl[ -\frac{\mathrm{i}p_{\theta,0}^2}{2m_0\hbar} (t-t_0) \Biggr] \exp \left[ -\frac{\mathrm{i}}{\hbar} (t-t_0)  \int_{\mathbb{R}} \mathrm{d}m(\beta) \,  \exp \left( \frac{p_{\theta}'}{m_0 R} \beta  \right) \right] \\
        = \delta \left( p_\theta' -p_{\theta,0} \right) \exp \Biggl\{ -\frac{\mathrm{i}}{\hbar} (t-t_0) \Biggl[ \frac{p_{\theta,0}^2}{2m_0} + \int_{\mathbb{R}} \mathrm{d}m(\beta) \,  \exp \left( \frac{p_{\theta}'}{m_0 R} \beta  \right) \Biggr] \Biggr\}.
        \label{appendix: propagator}
    \end{multline}
    solves the Schr\"{o}dinger equation
\begin{multline}
    \mathrm{i}\hbar \frac{\partial}{\partial t} K(p_\theta', t | p_{\theta,0}, t_0) \\
    = \frac{(p_\theta')^2}{2m_0} K(p_\theta', t | p_{\theta,0}, t_0) + \int_{-\infty}^\infty \mathrm{d}p_{\theta,1} W(p_\theta' - p_{\theta,1})K(p_{\theta,1}, t | p_{\theta,0}, t_0),
    \label{appendix: schro}
\end{multline}
where
\begin{equation}
    W(p_\theta' - p_{\theta,1}) = \frac{1}{2\pi} \int_{0}^{2\pi} \mathrm{d}\theta \, \mathrm{e}^{-\mathrm{i}(p_\theta' - p_{\theta,1})\theta/\hbar} V(\dot{\theta}(s)),
\end{equation}
and
\begin{equation}
    V(\dot{\theta}(s)) = \int_{\mathbb{R}} \mathrm{e}^{\beta \dot{\theta}(s)} \, \mathrm{d}m(\beta)
\end{equation}
with the initial condition 
\begin{equation}
		\lim_{t\searrow t_0} K(p_\theta', t | p_{\theta,0}, t_0) = \delta(p_\theta'-p_{\theta,0}).
	\end{equation}
\end{rem}
\begin{proof}
    With $p_\theta' = p_{\theta,0} = p_{\theta,1}$, we have \[W(p_\theta' - p_{\theta,1}) = \frac{1}{2\pi } \int_{0}^{2\pi} \mathrm{d}\theta \, V(\dot{\theta}(s))  
    = \int_{\mathbb{R}} \exp\left({\frac{\beta }{m_0 R}} p_{\theta,1}\right) \, \mathrm{d}m(\beta) \]
    and so
    \begin{multline*}
        \int_{-\infty}^\infty \mathrm{d}p_{\theta,1} W(p_\theta' - p_{\theta,1})K(p_{\theta,1}, t | p_{\theta,0}, t_0) \\
        = \int_{-\infty}^\infty \mathrm{d}p_{\theta,1} \int_{\mathbb{R}} \exp \left(\frac{\beta}{m_0R}p_{\theta,1} \right) \, \mathrm{d}m(\beta) K(p_{\theta,1}, t | p_{\theta,0}, t_0) \\
        = \int_{-\infty}^\infty \mathrm{d}p_{\theta,1} \left[ \int_{\mathbb{R}}  \exp \left(\frac{\beta}{m_0R}p_{\theta,1} \right) \, \mathrm{d}m(\beta) \right] \\
        \times \delta \left( p_{\theta,1} -p_{\theta,0} \right) \exp \Biggl\{ -\frac{\mathrm{i}}{\hbar} (t-t_0) \Biggl[ \frac{p_{\theta,0}^2}{2m_0} + \int_{\mathbb{R}} \mathrm{d}m(\beta) \,  \exp \left( \frac{p_{\theta,1}}{m_0 R} \beta  \right) \Biggr] \Biggr\} \\
        = \left[ \int_{\mathbb{R}}  \exp \left(\frac{\beta}{m_0R}p_{\theta,0} \right) \, \mathrm{d}m(\beta) \right] \\
        \times \exp \Biggl\{ -\frac{\mathrm{i}}{\hbar} (t-t_0) \Biggl[ \frac{p_{\theta,0}^2}{2m_0} + \int_{\mathbb{R}} \mathrm{d}m(\beta) \,  \exp \left( \frac{p_{\theta,0}}{m_0 R} \beta  \right) \Biggr] \Biggr\}. \\
    \end{multline*}
It would then be straightforward to show that the Schr\"{o}dinger equation, Eq.~\eqref{appendix: schro}, is satisfied by the propagator.
\end{proof}

\subsection{Bound State AB Effect in terms of the Perturbation of the Particle on a Circle}

Here, we consider in the WNA framework a particular case treated in Subsection~\ref{subsec: perturbed} when we specify the measure $m$. We obtain the propagator for the bound state AB effect starting from the case of the particle on a circle perturbed by a potential that is exponentially growing in angular velocity, Eq.~\eqref{appendix: propagator}.

A well-known fact as confirmed by experiments (see, e.g., the theoretical description in Refs.~\cite{aharonov1959significance, peshkin1990aharonov} and some confirming experiments in Refs.~\cite{peng2010aharonov, xiu2011manipulating}) is that the AB effect is undetectable for a magnetic flux $\phi$ that is integer-multiple of $2\pi\hbar c/ e$ (called the London's unit).  In other words, it is detectable only for a fractional part of the London's unit 
as units of the magnetic flux. Therefore, we are motivated to attempt to modify the potential $V$ to fit into the description of the bound state AB effect.  Incorporating these values, we have
\begin{equation}
    \alpha = \frac{ke}{2\pi \hbar c} \phi = \frac{ke}{2\pi \hbar c} \frac{2\pi \hbar c}{ne} = \frac{k}{n}, \quad k,n\in \mathbb{Z}\setminus\{0\},
    \label{appendix: experimentalvalues}
\end{equation}
where $k$ and $n$ has been chosen so as to conform with the fact that a charge is always a multiple of the elementary charge $e$.  More particularly, the ratio of two charges, say $q_1 = ke$ and $q_2 =ne$, 
must be $q_1 / q_2 = k/n$ in order to obtain a fraction. 
\begin{defn} \label{def:ABdetectoperator}
    The AB effect is detectable if $({k}/{n}) \in \mathbb{Q} \setminus \mathbb{Z}$, otherwise, if $({k}/{n}) \in \mathbb{Z}$, then the AB effect is not detectable.
\end{defn}

Let us now consider the following measure:
\begin{equation}
    m := g\delta_{b\alpha}, \quad g, b \in \mathbb{R}
    \label{dm}
\end{equation}
where $g$ is a coupling constant and $b$ is some parameter, 
so that we have chosen a measure concentrated on the magnetic flux parameter ``$\alpha$" of the AB effect multiplied by $b$. 

Substituting Eq.~\eqref{dm}
into the propagator, Eq.~\eqref{appendix: propagator}, we have obtained a propagator in the polar conjugate momentum space for a particle on a circle perturbed by an exponentially growing potential containing the magnetic term $\alpha$:
\begin{multline}
        K_V(p_\theta', t | p_{\theta,0}, t_0)         = \delta ( p_\theta' -p_{\theta,0} ) \exp \Biggl\{ -\frac{\mathrm{i}}{\hbar} (t-t_0) \Biggl[ \frac{p_{\theta,0}^2}{2m_0} + g\exp \left( b\alpha \frac{p_{\theta}'}{m_0 R}  \right) \Biggr] \Biggr\}.
        \label{appendix: propagator2}
    \end{multline}

With $\alpha$ given in Eq.~\eqref{appendix: experimentalvalues}, we have the potential $V$ of the form (let $\dot{\theta}' := p_\theta' / (m_0 R)$)
    \begin{multline}
        V = g\sum_{j = 0}^\infty \frac{(-1)^j}{j!} \left( \frac{bk}{n} \frac{p_{\theta}'}{m_0 R}  \right)^j = g\sum_{j = 0}^\infty \frac{(-1)^j}{j!} \left( \frac{bk}{n} \dot{\theta}'  \right)^j \\
        = g\Biggl[1  - \left( \frac{bk}{n} \right) \dot{\theta}'  + \frac{1}{2}\left( \frac{bk}{n} \right)^2  \dot{\theta}'^2 - \frac{1}{6}\left( \frac{bk}{n} \right)^3  \dot{\theta}'^3 + \ldots \Biggr]. \\ 
 \label{Vwithvalues}
    \end{multline}
    The first order approximation of $V$ given in Eq.~(\ref{Vwithvalues})  is given by
    \begin{equation}
    V \approx g - \frac{gbk}{n}\dot{\theta}',
    \end{equation}
    and the polar conjugate momentum space propagator is given by
    \begin{multline}
        K_{V}(p_\theta', t | p_{\theta,0}, t_0) \\
        = \delta \left( p_\theta' -p_{\theta,0} \right) \exp \Biggl\{ -\frac{\mathrm{i}}{\hbar} (t-t_0) \Biggl[ \frac{p_{\theta,0}^2}{2m_0} +  g\sum_{j = 0}^\infty \frac{(-1)^j}{j!} \left( \frac{bk}{n} \frac{p_{\theta}'}{m_0 R}  \right)^j  \Biggr] \Biggr\} \\
        \approx 
        \delta \left( p_\theta' -p_{\theta,0} \right)  \exp \Biggl\{ -\frac{\mathrm{i}}{\hbar} (t-t_0) \Biggl[ g + \frac{p_{\theta,0}^2}{2m_0} -  \frac{gbk}{n}\frac{p_{\theta}'}{m_0 R}  \Biggr] \Biggr\}. \label{approxpropagator}
        \end{multline}
    This scenario can be deemed as a bound state AB effect potential when $g=1$ 
    and $b=(1+B^{-1})$ where $B=({k\hbar p_{\theta}'})/({nc}{m_0 R})$ for $p_{\theta}' \neq 0$.

We now state our alternative result in obtaining a bound state AB effect momentum space propagator:
\begin{thm}
If the potential $V$ is given by
\begin{equation} \label{Vwithghat}
    V(\theta) = g\exp \left( \frac{bk\hbar}{n c} \dot{\theta} \right), \quad\quad ({k}/{n}) \in \mathbb{Q} \setminus \mathbb{Z}
\end{equation}
then $V$ will lead to an approximation of the bound state AB effect potential with the appropriate values of $g$ and $b$.
 The conjugate momentum space propagator is given by Eq.~\eqref{approxpropagator}.
\end{thm}

\section{Conclusion and Discussion}\label{sec:conclusions}
In this paper, we have demonstrated that within the WNA framework, ignoring the winding, a charged particle in a bound state AB potential maintains the conservation of the magnitude of the conjugate momentum. This phenomenon can be attributed to the conservation of the total angular momentum of the system $\Vec{L}$. Classically, $\Vec{L} = \Vec{r} \times \Vec{p}$, and since the magnitude of $\Vec{r}$ is constant, the magnitude of $\Vec{p}$ must also remain constant to preserve $\Vec{L}$.  
This is also true for a particle on a circle without potential energy and no winding; for example, refer to Eq.~\eqref{freeparticleinacircle}.  

Moreover, it might initially appear that, when considering the possibility for winding (with its associated quantized energy levels), quantum mechanics (e.g., Eq.~\eqref{TI2}) suggests that 
\begin{itemize}
    \item the quantization of energy, reliant on the integer $l$ (also seen in Eq.~\eqref{TI2}, the propagator in momentum space), suggests that the magnitude of the conjugate momentum ($p_\theta$) may not generally be conserved; and
    \item if we apply a ``driving force", e.g., a centripetal force, to preserve $p_\theta$ or if $p_\theta$ remains conserved (either continuously or at specific moments), then Eq.~\eqref{TI2} simplifies to Eq.~\eqref{TI}, signifying the classical conservation of angular momentum. Within these time intervals, Eq.~\eqref{TI2} also indicates that the winding does not directly influence the likelihood of detecting the charged particle with momentum $p_\theta'$ at a later time $t > t_0$ in the presence of the AB effect potential $V_{\mathrm{AB}}$.
\end{itemize}

The nature of the AB effect is such that, in cases like the bound state with a fixed radius and a constant magnetic flux, 
the total angular momentum is conserved; as stated in, e.g., Ref.~\cite{peshkin1990aharonov}. In the WNA framework, we have shown that, despite the presence of the potential $V_{\mathrm{AB}}$, the magnitude of the conjugate momentum associated with the angular position remains conserved (See Eq.~\eqref{TI}). The AB effect, therefore, does not alter the system's ``general equilibrium'' (total angular momentum) as depicted in Fig.~\ref{solenoid1}. Instead, it only affects the probability through the magnetic flux term $2\alpha /R := e\phi/(\pi \hbar c R)$, as shown in Eq.~\eqref{TI}.  Moreover, we have also shown that, even when considering winding, the condition for the propagator (Eq.~\eqref{TI2}) to satisfy the corresponding Schr\"{o}dinger equation is that the magnitude of the conjugate momentum must remain conserved.

However, one question that remains is: Why were we unable to incorporate the winding number information into the momentum space propagator?  The explanation is rooted in the Heisenberg uncertainty principle; in this scenario, determining the conjugate momentum of a charged particle at a specific time results in the loss of information about its precise position, including the winding number.  Using WNA methods, our findings thus verify the physics of the bound state AB effect with constant radius and magnetic flux.  These results can serve as a springboard for the design of experiments to further analyze the bound state AB effect.

Driven by the incomplete exploration of the AB effect and thus the potentiality to examine it from a different perspective, we used another method within the WNA framework, detailed in Sec.~\ref{app:alternative}. This alternative method entails a mathematical model capable of generating an approximation to the AB effect in a bound state, provided that the appropriate parameters are applied.

\section*{Acknowledgment}
ARN would like to thank the Department of Science and Technology - Science Education Institute (DOST-SEI) of the Philippine Government for financial support in the form of a graduate scholarship.  ARN would also like to thank the Indonesian Government for granting the one-year research stay at the Sanata Dharma University, Yogyakarta, where this collaborative research work with its Mathematics Department was done.
JLS was partially supported by the Center for Research in Mathematics and Applications (CIMA) related to the Statistics, Stochastic Processes
and Applications (SSPA) group, through the grant UIDB/MAT/04674/2020
of FCT--Funda{\c c\~a}o para a Ci{\^e}ncia e a Tecnologia, Portugal.

\bibliographystyle{plain}

\end{document}